\documentclass[11pt]{article}
\usepackage[margin=1.1in]{geometry}
\usepackage[onehalfspacing]{setspace}

\usepackage{amsmath}
\usepackage{amssymb}
\usepackage{amsthm}
\theoremstyle{plain} 

\newtheorem{assumptionR}{Assumption}

\newtheorem{assumptionS}{Assumption}

\newtheorem{lemmaC}{Lemma}

\newtheorem{theorem}{Theorem}

\newtheorem{corollary}{Corollary}
\newtheorem{corollaryS}{Corollary}

\newtheorem{proposition}{Proposition}

\theoremstyle{definition} 
\newtheorem{definition}{Definition}
\newtheorem{remark}{Remark}
\newtheorem{example}{Example}

\theoremstyle{remark} 

\usepackage{mathrsfs}
\usepackage{mlmodern}
\let\int\undefined
\DeclareSymbolFont{otherlargesymbols}{OMX}{cmex}{m}{n}
\DeclareMathSymbol{\int}{\mathop}{otherlargesymbols}{"52}
\let\oldint\int
\renewcommand{\int}{\oldint\nolimits}
\renewcommand{\iint}{\int\!\!\!\int}
\renewcommand{\iiint}{\int\!\!\!\int\!\!\!\int}
\let\sum\undefined
\DeclareSymbolFont{otherlargesymbols}{OMX}{cmex}{m}{n}
\DeclareMathSymbol{\sum}{\mathop}{otherlargesymbols}{"50}

\DeclareFontFamily{U}{mathx}{} 
\DeclareFontShape{U}{mathx}{m}{n}{ <-> mathx10 }{}
\DeclareSymbolFont{mathx}{U}{mathx}{m}{n}
\DeclareFontSubstitution{U}{mathx}{m}{n}
\DeclareMathAccent{\widecheck}{0}{mathx}{"71}

\usepackage[english]{babel}
\usepackage[T1]{fontenc}
\usepackage[bottom]{footmisc}

\usepackage{authblk}
\usepackage[title]{appendix}

\usepackage{titletoc}
\usepackage{tocloft}

\usepackage{natbib}
\usepackage[dvipsnames]{xcolor}

\usepackage{threeparttable}
\usepackage{booktabs}

\usepackage{graphicx}
\usepackage[position=bottom]{subfig}

\usepackage[colorlinks=true,citecolor=Blue,linkcolor=RubineRed,backref=page]{hyperref} %

\newcommand{\R}{\mathbb{R}}
\newcommand{\E}{\mathbb{E}}
\newcommand{\EE}[2]{\mathbb{E}#1[ #2 #1]}

\newcommand{\EEc}[3]{\mathbb{E}#1[ #2 #1| #3 #1]}
\newcommand{\pb}{\mathbb{P}}
\newcommand{\pbb}[2]{\mathbb{P}#1[ #2 #1]}
\newcommand{\pbc}[3]{\mathbb{P}#1[ #2 #1| #3 #1]}
\newcommand{\I}[1]{\text{\usefont{U}{bbold}{m}{n}1}\{ #1 \}}

\newcommand{\indep}{\perp \!\!\! \perp}
\newcommand{\eq}[2]{\overset{(#1)}{#2}} 
\newcommand{\supp}[1]{\mathcal{S}_{#1}}

\newcommand{\maintitle}{Structural Representations and Identification of Marginal Policy Effects}
\newcommand{\email}{author@example.org}

\newcommand\ToC{%
	\startcontents
	\setcounter{tocdepth}{2}
	\onehalfspacing
	\renewcommand{\cftsecfont}{\normalfont\scshape}
	\renewcommand{\cftsecleader}{\cftdotfill{\cftdotsep}}  
	\centering
	\printcontents{}{1}{{\scshape\large Contents \\} }
}

\begin{document}
\allowdisplaybreaks
\begin{titlepage}
	\thispagestyle{empty}
	\hypersetup{linkcolor=black}

	\title{ \vspace{-1.3cm}  \Large\bf
		\maintitle\footnote{This version replaces an earlier draft titled ``\textit{Structural Characterizations of Marginal Policy Effects}'' (\href{https://arxiv.org/abs/2506.11694}{https://arxiv.org/abs/2506.11694}).}
		}
	
	\author[a]{\normalsize\scshape Zhixin Wang\thanks{
			Corresponding author. E-mail: \href{mailto:wangzhixin.ok@qq.com}{\texttt{wangzhixin.ok@qq.com}}.
			
			\quad
			\href{mailto:zhangyu@suibe.edu.cn}{\texttt{zhangyu@suibe.edu.cn}} (Y. Zhang);\,
			\href{mailto:zy.zhang@mail.shufe.edu.cn}{\texttt{zy.zhang@mail.shufe.edu.cn}} (Z. Zhang).
			}}
	\affil[a,c]{School of Economics, Shanghai University of Finance and Economics}
	
	\author[b]{\normalsize\textsc{Yu Zhang}}
	\affil[b]{School of International Business, Shanghai University of International Business and Economics}
	
	\author[c]{\normalsize\textsc{Zhengyu Zhang}}

	\date{}

	\maketitle

	\begin{abstract}
		This paper investigates the structural interpretation of the marginal policy effect (MPE) within nonseparable models. We demonstrate that, for a smooth functional of the outcome distribution, the MPE equals its functional derivative evaluated at the outcome-conditioned weighted average structural derivative. This equivalence is definitional rather than identification-based. Building on this theoretical result, we propose an alternative identification strategy for the MPE that complements existing methods.

		\vspace{0.5em}
		 
	\end{abstract}

\thispagestyle{empty}

%
%
\end{titlepage}

\section{Introduction}
This paper investigates the structural interpretation of the \textit{marginal policy effect} (MPE) within nonseparable models. The MPE quantifies the marginal impact of a counterfactual change in the distribution of a policy variable on specific features of the outcome distribution. While the identification of the MPE has been studied by \citet{firpo2009}, \citet{rothe2010el,rothe2012} and \citet{mms2024}, its general connection to underlying structural effects remains underexplored. Specifically, \citet{firpo2009} and \citet{rothe2010el} examine this issue for the quantile functional.

This paper contributes to the literature by providing a structural interpretation of the MPE for general smooth functionals, thereby extending the analysis beyond the specific case of quantiles. Our structural analysis commences with the definition of the parameter itself, as opposed to its identification results.
We demonstrate that, under regularity conditions, the MPE for a Hadamard-differentiable functional equals its functional derivative, evaluated at the outcome-conditioned weighted average structural derivative.
For the quantile functional, we derive a structural decomposition of \textit{unconditional quantile regression} (UQR) estimand defined by \citet{firpo2009}. We further extend their proposition by showing that, under conditional independence and without requiring monotonicity in the structural error, the \(\tau\)-th UQR estimand equals the average structural derivative for individuals at the \(\tau\)-th quantile of the outcome distribution.

As a second contribution, this paper presents an alternative strategy for identifying the MPE. This strategy is applicable when the functional of interest is Hadamard-differentiable and the \textit{local average structural derivative}, as proposed by \cite{hm2007,hm2009}, is identifiable.
We present two general identification propositions: one for the single-function model under conditional independence and the other for the triangular system using the control variable method of \citet{imbensNewey2009}. These propositions can be used to derive or generalize existing identification results for MPEs.

The literature has extensively studied the structural interpretations of parameters and estimands. For instance, \citet{cherno2013} provide treatment effect interpretations for counterfactual distribution effects under unconfoundedness. \citet{sasaki2015} examines structural interpretations of conditional quantile regression estimands. The relationship between continuous quantile treatment effects and structural partial derivatives is also explored by \citet{su2019} for single-equation models, and by \citet{chernoPanel2015} for panel data models.

This paper is organized as follows. Section \ref{sec:para} defines the target parameter. Section \ref{sec:struc.interpre} introduces the structural model and presents the main results. Section \ref{sec:apply} presents the general identification propositions and their applications. Section \ref{sec:discuss.esti} discusses the estimation approaches briefly. Section \ref{sec:conclusion} concludes. Complete proofs of all results are available in the Supplementary Appendix.

\section{Target Parameter}\label{sec:para}
Let \(Y\) denote the outcome variable and \(D\) the policy variable subject to intervention, with both variables taking values in $\R$. 
\begin{definition}[Policy Intervention and Counterfactual Outcome]
	Let $\pi_t:\R \to \R$ denote a \textit{policy function} indexed by \(t\in[0,1]\). 
	Policymakers can implement a marginal intervention through \(\pi_t\), perturbing the policy variable \(D\) to create the counterfactual variable \(D^t:=\pi_t(D)\). The \textit{counterfactual outcome} corresponding to \(D^t\) is denoted by \(Y^t\). This outcome remains unobserved because the structural relationship between \(D^t\)  and \(Y^t\) is unspecified.
\end{definition}
The above policy intervention corresponds to a counterfactual change in the distribution of the policy variable, from its original cumulative distribution function (CDF) $F_D$ to $F_{D^t}$. For instance, an intervention implemented by $\pi_t(D)=D+t$ yields the counterfactual CDF $F_{D^t}(d)=\pb[D+t\leq d]=F_D(d-t)$. Unlike the counterfactual change considered here,  standard treatment effect models define potential outcomes \(Y(d)\) and evaluate the impact of changing the value of \(D\) from $d$ to $d'$ on functionals of $F_{Y(d)}$, independent of changes in the distribution of \(D\) itself.  
To further clarify the distinction, consider a linear random coefficient model \(Y(d)=\beta_0(\varepsilon) + \beta_1(\varepsilon)d\) with \(Y=Y(D)\) and \(Y^t=Y(D^t)\), where \(\varepsilon\) captures unobserved heterogeneity. In this context, the average policy effect is \(\E[Y^t-Y]=\E[(\pi_t(D)-D)\beta_1(\varepsilon)]\), while the average treatment effect is \(\E[Y(d')-Y(d)]=(d'-d)\E[\beta_1(\varepsilon)]\).

We now present several concrete examples of policy functions \(\pi_t\).
\begin{example}
	\citet{firpo2009} and \citet{rothe2010el} study the location shift intervention \(\pi_t(D)=D+t\), which represents a special case of the more general \(\pi_t(D)=\mu + l(t) + (D-\mu)s(t)\) proposed by \citet{mms2024}. Here, \(\mu\) is a known parameter, while \(l(t)\) and \(s(t)>0\) represent the location and scale transformations, respectively.  
	\citet{rothe2012} analyzes the rank-preserving transformation $\pi_t(D)= H_t^{-1}(F_D(D))$, where \(H_t^{-1}(\tau):=\inf\{d:H_t(d) \geq \tau\}\) is the \(\tau\)-th quantile of a perturbed distribution \(H_t\). Other examples include the mean-preserving transformation \(\pi_t(D) = \E[D] + (1+\alpha t)(D-\E[D])\), where \(\alpha \in [-1,1]\) governs the extent of dispersion.
\end{example}

Let \(\mathcal{F}\) denote the space of one-dimensional CDFs and \(\ell^\infty(\mathbb{T})\) the space of all bounded functions defined on a set \(\mathbb{T}\). Let \(F_{Y^t}(\cdot):=\pb[Y^t\leq \cdot]\) denote the CDF of the counterfactual outcome \(Y^t\). We study the following parameter.
\begin{definition}\label{def:mpe}
	For a functional $\Gamma:\mathcal{F} \to \ell^\infty(\R)$, the \textit{Marginal Policy Effect} (MPE) of a counterfactual change in $D$ on $\Gamma(F_Y)$ is defined as
	\begin{align*}
		\theta_\Gamma:=\partial_t\Gamma(F_{Y^t})|_{t=0} :=\lim_{t \downarrow 0} \frac{\Gamma(F_{Y^t})-\Gamma(F_Y)}{t}  
	\end{align*}
	provided that the limit exists. 
\end{definition}

We illustrate the above parameter through several representative examples of functionals \(\Gamma\).
\begin{example}\label{exp:gamma.para}
	A theoretically useful special case is the identity mapping \(id_y:= F\mapsto F(y)\), which yields the distributional MPE:
	\begin{align*}
		\theta_{id}(y):=\lim_{t\downarrow 0} \frac{F_{Y^t}(y)-F_Y(y)}{t}.	
	\end{align*}
	Another fundamental example is the quantile functional \(Q_\tau:= F \mapsto \inf\{y:F(y) \geq \tau\}\) for \(\tau \in (0,1)\), which generates the quantile MPE studied by \citet{firpo2009} and \citet{mms2024}:
	\begin{align*}
		\theta_{Q}(\tau):= \lim_{t \downarrow 0} \frac{Q_\tau(F_{Y^t})-Q_\tau(F_Y)}{t}.
	\end{align*}
	For the mean functional \(\mu: F \mapsto \int y \,dF(y)\), the average MPE is defined as
	\begin{align*}
		\theta_\mu := \lim_{t \downarrow 0}\frac{\mu(F_{Y^t})-\mu(F_Y)}{t}. 
	\end{align*}
	Furthermore, \(\Gamma\) can be specified as various inequality measures. For instance, consider the \textit{Gini coefficient} 
	\(GC:F \mapsto 1 - 2\int_{0}^{1}L_p(F)\,dp\), where \(L_p(F):=\frac{\int_{0}^{p}Q_\tau(F)\,d\tau}{\int y\,dF(y)}\) denotes the \textit{Lorenz curve}. The Gini MPE is defined as 
	\begin{align*}
		\theta_{GC}:=\lim_{t \downarrow 0}\frac{GC(F_{Y^t}) - GC(F_{Y})}{t},
	\end{align*}
	which quantifies the marginal effect of a perturbation in \(F_D\) on the inequality of outcome $Y$.
\end{example}

\section{Structural Interpretations}\label{sec:struc.interpre}
\subsection{Nonseparable Models and Related Assumptions}
We introduce the following nonseparable model to characterize the dependence structure.
\begin{assumptionS} Let $m(\cdot,\cdot,\cdot)$ be a measurable function of unknown form. \label{A:struct.m.l}~
	\begin{itemize}
		\item[(a)] The status quo outcome $Y$ is generated by
		\begin{align}
			Y=m(D,X,\varepsilon), \label{eq:status.quo.Y}
		\end{align}
		where $D$ is the continuous policy variable, $X$ is a vector of covariates, and $\varepsilon$ represents the unobserved disturbance. Suppose that $m(d,x,e)$ is continuously differentiable in $d$ for each $(x,e)$. The partial derivative $\partial_d m(\cdot,\cdot,\cdot)$ is termed the structural derivative.
		
		\item[(b)] The counterfactual outcome $Y^t$ is generated by
		\begin{align}
			Y^t=m(\pi_t(D),X,\varepsilon)=:m(D^t,X,\varepsilon), \label{eq:counterfactual.Y}
		\end{align}
		where $\pi_0(d)\equiv d$ (status quo policy). For each $d$, assume $\pi_t(d)$ is continuously differentiable in $t$ at $t=0$, with bounded derivative $\dot{\pi}(d):=\partial_t\pi_t(d)|_{t=0}$.
	\end{itemize}
\end{assumptionS}

To derive the structural interpretation, the following conditions are required.
\begin{assumptionR}[Regularity]\label{A:regularity}~
	\begin{itemize}
		\item[(a)] The CDF $F_Y$ is continuously differentiable on its support \(\supp{Y}\), with density $f_Y$ satisfying \(0<f_Y(y)<\infty\) for every \(y\in \supp{Y}\).
		
		\item[(b)] Let $\dot{m}(d,x,e):=\partial_t m(\pi_t(d),x,e)|_{t=0}=\dot{\pi}(d)\partial_d m(d,x,e)$. Suppose that $\dot{m}(\cdot,\cdot,\cdot)$ is measurable and satisfies
		\begin{align*}
			\pbb{\Big}{\big| m(\pi_t(D),X,\varepsilon) - m(D,X,\varepsilon) - t \,\dot{m}(D,X,\varepsilon) \big| \geq \nu t } = o(t) 
		\end{align*}
		as $t \downarrow 0$ for every $\nu > 0$.
		
		\item[(c)] The distribution of $(Y,\dot{m}(D,X,\varepsilon))$ is absolutely continuous with density $f_{Y,\dot{m}}(y,y')$ that is continuous in $y$ for each $y'$. Furthermore, there exists a Lebesgue integrable function $g:\R \to \R$ satisfying $\int |y' g(y')|\,dy' < \infty$ and a positive constant \(C\) such that $f_{Y,\dot{m}}(y,y') \leq C |g(y')|$ for all \((y,y')\).
	\end{itemize}
\end{assumptionR}
Assumption \ref{A:regularity} does not impose structural conditions like the invariance of \(F_{Y|D,X}\) under small manipulations, as assumed in \citet{firpo2009}. Assumption \ref{A:regularity} (a) implies that $Y$ is continuously distributed with a regular density. Assumptions \ref{A:regularity}(b)--(c) are analogous to \citeauthor{hm2007}'s (\citeyear{hm2007}) Assumptions A3--A4, implying smooth variation in the structural response $\dot{m}$.

The directional derivative structure of \(\theta_\Gamma\) motivates our focus on smooth functionals \(\Gamma\). 
\textit{Hadamard-differentiable} functionals encompass key causal statistics for policy analysis, including means, quantiles, and common inequality measures (e.g., the Gini coefficient).\footnote{For additional Hadamard-differentiable inequality measures, see, e.g., \citet{rothe2010joe} and \citet{firpo2016}.} This differentiability is necessary for both the functional Delta method (\citealp{vaart2000}) and standard bootstrap inference (\citealp{fang2019}).

Formally, let $(\mathbb{D},\|\cdot\|_{\mathbb{D}})$ and $(\mathbb{B},\|\cdot\|_{\mathbb{B}})$ be normed spaces.
A functional $\Gamma:\mathbb{D}_\Gamma \subseteq \mathbb{D} \to \mathbb{B}$ is \textit{Hadamard-differentiable} at $F\in\mathbb{D}_\Gamma$ tangentially to a set $\mathbb{D}_0 \subseteq \mathbb{D}$, if there exists a continuous linear functional $\Gamma'_{F}: \mathbb{D}_0 \to \mathbb{B}$ such that
\begin{align*}
	\lim_{t \downarrow 0} \left|\left| \frac{\Gamma(F + t h_t) - \Gamma(F)}{t} -\Gamma'_{F}(h)\right|\right|_{\mathbb{B}} = 0
\end{align*}
for every sequence $\{h_t\}\subset \mathbb{D}$ satisfying $h_t \to h\in\mathbb{D}_0$ and $F+th_t\in \mathbb{D}_\Gamma$, where $\Gamma'_F$ is called the \textit{Hadamard derivative} of $\Gamma$ at $F$. See, for example,  \citet[Chapter 20.2]{vaart2000} for more details.

\subsection{Structural Representation of the MPE}
\begin{theorem}[Structural Representation of $\theta_\Gamma$]\label{thm:struct.theta}
	Under Assumptions \ref{A:struct.m.l} and \ref{A:regularity}, we obtain
	\begin{align*}
		\theta_{id}(y) :=\lim_{t\downarrow 0} \frac{F_{Y^t}(y)-F_Y(y)}{t}
		= \EEc{\big}{\omega^f(y,D)\,\partial_d m(D,X,\varepsilon)}{Y=y}
	\end{align*}
	for every $y\in \supp{Y}$, where $\omega^f(y,d):=-f_Y(y)\dot{\pi}(d)$.
	If $\Gamma$ is Hadamard-differentiable at $F_Y$, we have
	\begin{align*}
		\theta_\Gamma
		= \Gamma'_{F_Y} \Big( \EEc{\big}{\omega^f(\cdot,D)\,\partial_d m(D,X,\varepsilon)}{Y=\cdot}\Big), 
	\end{align*}
	where \(\Gamma'_{F_Y}\) denotes the Hadamard derivative of \(\Gamma\) at \(F_Y\).
\end{theorem}
Theorem \ref{thm:struct.theta} states that: (i) the marginal effect of counterfactual changes in \(F_D\) on the outcome CDF \(F_Y(y)\), equals the weighted average structural effect for individuals with \(Y=y\). (ii) For any Hadamard-differentiable transformation \(\Gamma\), the marginal impact on \(\Gamma(F_Y)\) is given by the Hadamard derivative of \(\Gamma\) at \(F_Y\), applied to the weighted average structural effect function identified in (i).

Theorem \ref{thm:struct.theta} extends the literature in two principal ways: (i) When considering \(\pi_t(\cdot) = H^{-1}_t(F_D(\cdot))\), where \(H_t\) is the counterfactual CDF of \(F_D\), Theorem \ref{thm:struct.theta} offers a structural interpretation of \citeauthor{rothe2012}'s (\citeyear{rothe2012}) \textit{marginal partial distributional policy effect}. (ii) For \(\Gamma=Q_\tau\), it provides a structural interpretation for the \textit{unconditional quantile effect} studied by \citet{mms2024}, as formalized below.
\begin{corollary}[Structural Representation of $\theta_Q$] \label{cory:struct.theta.Q}
	For \(\Gamma=Q_\tau\), under the assumptions of Theorem \ref{thm:struct.theta} and assuming the support \(\supp{Y}\) is compact,  we obtain
	\begin{align*}
		\theta_Q(\tau)=\EEc{\big}{\dot{\pi}(D)\,\partial_d m(D,X,\varepsilon)}{Y=q_\tau}
	\end{align*}
	for every \(\tau \in (0,1)\), where \(q_\tau := Q_Y(\tau)\).
\end{corollary}
Corollary \ref{cory:struct.theta.Q} states that the MPE of a counterfactual change in \(F_D\) on the \(\tau\)-th quantile of the outcome distribution, \(Q_\tau(F_Y)\), equals the weighted average structural derivative for individuals with outcome value \(q_\tau\), with weights determined by the policy function \(\pi_t\). 

\subsection{Structural Representation of the UQR}\label{sec:struct.uqr}
This section extends the structural interpretation of the \textit{unconditional quantile partial effect} (UQPE) proposed by \citet{firpo2009}. 
For comparison purposes, we focus on the location shift \(D^t=D+t\) and denote \(\theta_Q\) as \(\theta_Q^L\) for this specific case. 
The UQPE of \citet[p. 958]{firpo2009} is defined as
\begin{align*}
	\beta^{\mathrm{UQR}}(\tau)
	&:=-\frac{\EE{\big}{\partial_d{F_{Y|D,X}(q_\tau|D,X)}}}{f_Y(q_\tau)}.
\end{align*}
We denote this UQPE as \(\beta^{\mathrm{UQR}}(\tau)\), which corresponds to their \textit{unconditional quantile regression} (UQR) estimand rather than an unobserved parameter.

Note that \(\theta_Q^L\) is not equal to  \(\beta^{\mathrm{UQR}}\) under the standard regularity assumptions mentioned above. Consequently, Corollary \ref{cory:struct.theta.Q} cannot be used to structurally interpret  \(\beta^{\mathrm{UQR}}\) as \( \E[\partial_dm(D,X,\varepsilon)|Y=q_\tau]\) under these assumptions.   \citet{firpo2009}, however, demonstrate that \(\theta_Q^L(\tau)=\beta^{\mathrm{UQR}}(\tau)\) holds under the structural assumption \(F_{Y^t|D+t,X} = F_{Y|D,X}\). 
A broader question thus emerges: What does the UQR estimand $\beta^{\mathrm{UQR}}$ identify in nonseparable models when structural assumptions are not imposed? To address this question formally, we introduce the following conditions.
\begin{assumptionR}[Regularity] \label{A:regular.hm}
	For each $(d,x)\in\supp{D}\times\supp{X}$:
	\begin{itemize}
		\item[(a)] The conditional distribution of $Y$ given $(D,X)$ is absolutely continuous, with density $f_{Y|D,X}(\cdot|d,x)$ that is strictly positive, bounded and continuous on its support. Furthermore, \(F_{Y|D,X}(y|d,x)\) is continuously differentiable in \(d\) for each $(y,x)$.
		
		\item[(b)] $\partial_d m(d,x,\cdot)$ is measurable and satisfies
		\[
		\pbc{\Big}{\big| m(d+\delta,x,\varepsilon) - m(d,x,\varepsilon) - \delta\,\partial_d m(d,x,\varepsilon) \big| \geq \nu \delta}{D=d,X=x} = o(\delta)
		\]
		as $\delta \downarrow 0$ for every $\nu > 0$.
		\item[(c)] The conditional distribution of \((Y,\partial_d m(d,x,\varepsilon))\) given \((D,X)\) is absolutely continuous, with density $f_{Y,\partial_dm^{d,x}|D,X}(\cdot,y'|d,x)$ that is continuous for each $(y',d,x)$. Furthermore, there exists a Lebesgue integrable function $g:\R\to\R$ satisfying $\int |y'g(y')|\,dy'<\infty$ and a positive constant $C$ such that $f_{Y,\partial_dm^{d,x}|D,X}(y,y'|d,x) \leq C |g(y')|$ for every \((y,y',d,x)\).
		
		\item[(d)] $\pb[m(d',x,\varepsilon)=Q_{Y|D,X}(\alpha|d,x)|D=d,X=x]=0$ for $d'$ in a neighborhood of $d$. The conditional distribution of \(\varepsilon\) given \((D,X)\) is absolutely continuous with the strictly positive density $f_{\varepsilon|D,X}(\cdot|d,x)$. Furthermore, $f_{\varepsilon|D,X}(e|d,x)$ is partially differentiable in \(d\) for each \((e,x)\), and satisfies
		\(\sup_{|\delta| \leq c} \left| \frac{f_{\varepsilon|D,X}(e|d+\delta,x)-f_{\varepsilon|D,X}(e|d,x)}{\delta}\right| < \infty\) 
		for some positive constant $c$.
	\end{itemize}
\end{assumptionR}
Assumptions \ref{A:regular.hm}(a)--(b) correspond to \citeauthor{hm2007}'s (\citeyear{hm2007}) Assumptions A2--A3, where the continuous differentiability of \(F_{Y|D,X}(y|\cdot,x)\) is equivalent to that of \(Q_{Y|D,X}(\alpha|\cdot,x)\). Assumptions \ref{A:regular.hm}(c) parallels Assumption A4 in \citet{hm2007}, ensuring smooth variation in the structural effect \(\partial_dm\). 
Finally, Assumption \ref{A:regular.hm}(d) aligns with Assumption A1 in \citet{hm2009}, which facilitates decomposition of the endogeneity-induced bias in the UQR estimand.

\begin{theorem}[Structural Representation of \(\beta^{\mathrm{UQR}}\)]\label{thm:struct.uqr.reg}
	Under Assumptions \ref{A:struct.m.l} and \ref{A:regular.hm}, we obtain
	\begin{align*}
		\beta^{\mathrm{UQR}}(\tau)
		=\EEc{\big}{\partial_dm(D,X,\varepsilon)}{Y=q_\tau} - \frac{\EE{\big}{\I{Y\leq q_\tau}\,\partial_d{\ln\big(f_{\varepsilon|D,X}(\varepsilon|D,X)\big)}}}{f_Y(q_\tau)} 
	\end{align*}
	for every $\tau \in (0,1)$.
\end{theorem}
Theorem \ref{thm:struct.uqr.reg} generalizes Proposition 1 of \citet{firpo2009}. Note that if conditional independence $\varepsilon \indep D|X$ are satisfied, we have \(\partial_d\ln(f_{\varepsilon|D,X}(\varepsilon|d,x)) = 0\) and  
\begin{align}
	\beta^{\mathrm{UQR}}(\tau)= \EEc{\big}{\partial_d{m(D,X,\varepsilon)}}{Y=q_\tau}, \label{eq:struct.uqr}
\end{align}
Therefore, under conditional independence, the estimand \(\beta^{\mathrm{UQR}}(\tau)\) is interpreted as the average structural effect for individuals with outcome value \(q_\tau\). Furthermore, if $e\mapsto m(d,x,e)$ is strictly monotonic, as assumed by \citeauthor{firpo2009}'s (\citeyear{firpo2009}) Proposition 1, we can derive their result using Equation (\ref{eq:struct.uqr}).

\begin{remark}
	For identifying \(\theta_Q^L\) via \(\theta_Q^L(\tau)=\beta^{\mathrm{UQR}}(\tau)\), the conditional independence assumption $\varepsilon \indep D|X$ is stronger than the distributional invariance condition \(F_{Y^t|D+t,X}=F_{Y|D,X}\) used by \citet{firpo2009}. This is because
	\[F_{Y^t|D+t,X}(\cdot|d,x)=\int \I{m(d,x,e)\leq \cdot}\,dF_{\varepsilon|D+t,X}(e|d,x),\]
	and  under $\varepsilon \indep D|X$, we have \(F_{\varepsilon|D+t,X}(e|d,x)=F_{\varepsilon|X}(e|x) = F_{\varepsilon|D,X}(e|d,x)\).
	Nonetheless, the conditional independence remains essential for the structural interpretation of $\beta^{\mathrm{UQR}}(\tau)$, as established in Theorem \ref{thm:struct.uqr.reg} or initially demonstrated by \citet{firpo2009}.
\end{remark}

To conclude this section, Table \ref{tb:struct.rep} summarizes the structural interpretations of the various MPEs. All reported results follow directly from Theorem \ref{thm:struct.theta} or Theorem \ref{thm:struct.uqr.reg}.
\begin{table}[htb]
	\centering
	\fontsize{10}{20}\selectfont
	\begin{threeparttable}
		\caption{Structural Representations of the MPEs}\label{tb:struct.rep}
		\begin{tabular}{l l l l}
			\hline
			Intervention & Functional & Structural Interpretation  & Related Literature\\
			\hline
			\(\pi_t(D)\) & \(id_y\) & \(\E[\omega^f(y,D)\cdot\partial_dm|Y=y]\) &  ------ \\
			\(D+t\) & \(Q_\tau\) & \(\beta^{\mathrm{UQR}}(\tau) =\E[\partial_dm|Y=q_\tau]\) under \(\varepsilon\indep D|X\)  & \citet{firpo2009}\\
			\(\pi_t(D)\) & \(Q_\tau\) & \(\E[\dot{\pi}(D)\cdot\partial_dm|Y=q_\tau]\) & \citet{mms2024}\\
			\(\pi_t(D)\) & \(\mu\) & \(\E[\dot{\pi}(D)\cdot\partial_dm]\)  & \citet{wooldridge2004}  \\
			\(D+t\) & \(\Gamma\)  & \(\Gamma'_{F_Y}\big(\E[\partial_dm|Y=\cdot]\big)\) &  \citet{rothe2010el}\\
			\(H^{-1}_t(F_D(D))\) & \(\Gamma\) & \(\Gamma'_{F_Y}\left(\EEc{\big}{\omega^H(D)\cdot\partial_dm}{Y=\cdot}\right)\) &  \citet{rothe2012}\\
			\(\pi_t(D)\) & \(GC\) & \( \EE{\big}{\omega^{GC}(Y,D)\cdot\partial_dm} \) &  ------ \\
			\hline
		\end{tabular}
		\begin{tablenotes}
			\fontsize{8}{12}\selectfont
			\item \textit{Notation:} \(\partial_dm=\partial_dm(D,X,\varepsilon)\); \(H_t=F_D + t(G_D-F_D)\); \(\omega^H(d)=\frac{F_D(d)-G_D(d)}{f_D(d)}\); \(\omega^{GC}(y,d):=\frac{2\varphi(F_Y(y))\dot{\pi}(d)}{\E[Y]^2}\) and \(\varphi(\tau) = \int_{0}^{1}(\tau-p)Q_p(F_Y)\,dp\).
		\end{tablenotes}
	\end{threeparttable}
\end{table}

\section{Applying Theorem \ref{thm:struct.theta} for Identification} \label{sec:apply}
Theorem \ref{thm:struct.theta} suggests a practical identification strategy for \(\theta_\Gamma\), which we term the \textit{LASD method}. This approach builds on the identification of the \textit{local average structural derivative} (LASD) proposed by \citet{hm2007}. This quantity is defined as
\[\theta_{\mathrm{LASD}}(d,x,y):=\EEc{\big}{\partial_d m(d,x,\varepsilon)}{D=d,X=x,Y=y}. \]
We note that, by the law of iterated expectations, Theorem \ref{thm:struct.theta} implies the following relationship:
\begin{align}
	\theta_\Gamma
	= \Gamma'_{F_Y} \Big( \EEc{\big}{\omega^f(\cdot,D)\,\theta_{\mathrm{LASD}}(D,X,\cdot)}{Y=\cdot}\Big), \label{eq:relation.theta.Gamma.LASD}
\end{align}
where the weight $\omega^f(y,d):=-f_Y(y)\dot{\pi}(d)$ is estimable from data. Consequently, identifying \(\theta_\Gamma\) reduces to identifying \(\theta_{\mathrm{LASD}}\).

This section establishes two general identification results for \(\theta_\Gamma\): one applicable to single-equation models and another for triangular systems. Prior to formal analysis, we compare the key identification step of the LASD method (Equation (\ref{eq:relation.theta.Gamma.LASD})) with existing approaches in the literature.

\subsection{Comparison of the MPE Identification Strategies}
\citet{firpo2009} and \citet{mms2024} employ \textit{influence function} (IF) approximations for identification purposes. 
\citet{firpo2009} analyze the counterfactual distribution \(F_{Y^t} =\int F_{Y|D,X}\,dF_{D+t,X}\), assuming distributional invariance \( F_{Y^t|D+t,X}=F_{Y|D,X}\). Their MPE identification builds on the G\^ateaux derivative characterization of the IF: 
\begin{align*}
	\theta_\Gamma
	= \int \mathrm{RIF}(y;\Gamma,F_Y)\,d(G_Y^*-F_Y)(y), 
\end{align*}
where \(\mathrm{RIF}(y;\Gamma,F_Y):=\Gamma(F_Y) + \mathrm{IF}(y;\Gamma,F_Y)\) denotes the \textit{recentered influence function} (RIF) and the distribution \(G^*_Y\) is constructed to satisfy the first-order local approximation \(F_{Y^t} = F_Y + t(G^*_Y-F_Y) + \mathcal{O}(t^2)\) as \(t \downarrow 0\). 
This methodology applies to \textit{G\^ateaux-differentiable}  functionals \(\Gamma\), a weaker condition than Hadamard differentiability, while specifically designed to evaluate counterfactual effects stemming solely from the location shift \(D^t=D+t\) in \(F_D\).

\citet{mms2024} examine general counterfactual transformations \(\pi_t(D)\) and employ the IF approximation of the Hadamard-differentiable functionals \(\Gamma\) for identification. Their approach builds on 
\begin{align*}
	\theta_\Gamma = \int \mathrm{IF}(y;\Gamma,F_Y)\,d\theta_{id}(y),
\end{align*}
where \(\theta_{id}\) is identified under the conditional independence. This methodology differs from the LASD approach in two key respects: (i) It uses the analytic form of the influence function \(\mathrm{IF}(\cdot;\Gamma,F_Y)\) rather than the Hadamard derivative \(\Gamma'_{F_Y}\). (ii) It requires different smoothness conditions than Assumption \ref{A:regular.hm}, including continuous differentiability of both \(f_{D,X}(\cdot,x)\) and \(f_{Y|D,X}(y|\cdot,x)\) for each \(x\in \supp{X}\) and \(y\in  \supp{Y} \).

Differing from the IF approximation approach, \citet{rothe2010el,rothe2012} establishes the identification building on the Hadamard derivative characterization
\begin{align*}
	\theta_\Gamma = \Gamma'_{F_Y}\big(\theta_{id}\big).
\end{align*}
While our LASD method also employs this form, the key distinction lies in the implementation: Rothe's method directly identifies \(\theta_\Gamma\) through \(\theta_{id}\) for specific marginal perturbations of \(F_D\) (e.g. \(D^t=D+t\) or \(H_t^{-1}(F_D(D))\)), whereas our approach first establishes the relationship between \(\theta_\Gamma\) and \(\theta_{\mathrm{LASD}}\) for general perturbations \(D^t=\pi_t(D)\), and then identifies \(\theta_\Gamma\) via \(\theta_{\mathrm{LASD}}\).  
This leads to the general identification formulas for \(\theta_\Gamma\) that include existing results as special cases, as we demonstrate below.

\subsection{Single-Equation Model}
\begin{proposition}\label{pron:id.theta.exo}
	Under the assumptions of Theorem \ref{thm:struct.theta}, Assumption \ref{A:regular.hm}(a)--(c), and the conditional independence $\varepsilon \indep D \,|X$, we obtain
	\begin{align*}
		\theta_\Gamma = \Gamma'_{F_Y} \Big(\EEc{\big}{\omega^f(\cdot,D) \beta(D,X,\cdot)}{Y=\cdot} \Big)
	\end{align*}
	for every Hadamard-differentiable functional \(\Gamma\), where \(\beta(d,x,y):=-\frac{\partial_d{F_{Y|D,X}(y|d,x)}}{f_{Y|D,X}(y|d,x)}\).
\end{proposition} 
Proposition \ref{pron:id.theta.exo} provides a general identification formula of \(\theta_\Gamma\) in the single-function model (\ref{eq:status.quo.Y}), where \(\beta(d,x,y)\) identifies \(\theta_{\mathrm{LASD}}(d,x,y)\) under the conditional independence (\citealp{hm2007}).

Turning to the quantile case, we define the conditional $\alpha$-quantile derivative (CQD) as $\beta^{\mathrm{CQD}}(\alpha,d,x):=\partial_d{Q_{Y|D,X}(\alpha|d,x)}$, corresponding to what \citet{firpo2009} call the conditional quantile partial effect (CQPE).
\begin{corollary} \label{cory:theta.Q.exo}
	For \(\Gamma=Q_\tau\), under the assumptions of Proposition \ref{pron:id.theta.exo} and assuming the support \(\supp{Y}\) is compact, we obtain
	\begin{align}
		\theta_Q(\tau)
		&= \frac{-1}{f_{Y}(q_\tau)}\int \dot{\pi}(d)\,\partial_d F_{Y|D,X}(q_\tau|d,x) \,dF_{D,X}(d,x) \label{eq:Q.mms}\\
		&= \int \dot{\pi}(d)\,\left(\omega_{q_\tau}(d,x)\cdot\beta^{\mathrm{CQD}}(\zeta_\tau(d,x),d,x)\right)\,dF_{D,X}(d,x)\label{eq:Q.ffl.pron2}\\
		&=\int \dot{\pi}(d)\,\beta^{\mathrm{CQD}}(\zeta_\tau(d,x),d,x)\,dF_{D,X|Y}(d,x|q_\tau), \label{eq:Q.alejo}
	\end{align}
	for every \(\tau \in (0,1)\), where \(\omega_y(d,x):=f_{Y|D,X}(y|d,x)/f_Y(y)\), and \(\zeta_\tau(d,x):=\{\alpha\in(0,1):Q_{Y|D,X}(\alpha|d,x)=q_\tau\}\) denotes the `matching' function.
\end{corollary}
Equation (\ref{eq:Q.mms}) coincides with Corollary 1 of \citet{mms2024}. Furthermore, under location shifts \(\pi_t(d)=d+t\) (where \(\dot{\pi}(d)\equiv 1\)), Equation (\ref{eq:Q.ffl.pron2}) recovers Proposition 1(ii) of \citet{firpo2009}, and Equation (\ref{eq:Q.alejo}) aligns with Lemma 1 of \citet{alejo2024}. 

Proposition  \ref{pron:id.theta.exo} can also be used to establish  \citeauthor{rothe2012}'s (\citeyear{rothe2012}) Theorem 4 under the additional regularity conditions introduced therein.
\begin{corollary}\label{cory:rothe2012}
	Let the assumptions of Proposition \ref{pron:id.theta.exo} hold. For \(D^t=H_t^{-1}(F_D(D))\), suppose that for each \( t\in [0,1]\), (a) \(\supp{D^t} \subseteq \supp{D}\); (b) \(H_t\) is continuous, and $f_D > 0$ on $\supp{D}$; (c) the copula \(C(\cdot,\cdot)\) of \((D^t,X)\) is partially differentiable with respect to its first component. Then:
	\begin{itemize}
		\item[(i)] For a marginal perturbation \(H_t(d) = F_D(d) + t\,(G_D(d) - F_D(d))\), where \(G_D\) denotes a CDF of \(D\), we have
		\(\dot{\pi}(d) = -(G_D(d)-F_D(d))/f_D(d)\), and 
		\begin{align*}
			\theta_\Gamma 
			&= \int \Gamma'_{F_Y} \big(F_{Y|D,X}(\cdot|d,x)\big)\,d\left(F_{X|D}(x|d)(G_D(d)-F_D(d))\right)  \\
			&= \int \Gamma'_{F_Y}\big(F_{Y|D,X}(\cdot|d,x)\big)\,d\left(\partial_t C\big(H_t(d),F_X(x)\big)|_{t=0}\right).
		\end{align*}
		
		\item[(ii)] For a location shift \(F_{D^t}(d)=F_{D}(d-t)\), we have \(\theta_\Gamma=\Gamma'_{F_Y}\big(\E[\partial_d F_{Y|D,X}(\cdot|D,X)]\big)\).
	\end{itemize} 
\end{corollary}
Note that \citet{rothe2012} adopts the conditional independence assumption \( \varepsilon \indep  U|X\) for identification, where \(U:=F_D(D)\) represents the rank of $D$ in its distribution. This condition is equivalent to \(\varepsilon \indep D|X\)  when the policy variable $D$ is continuously distributed.

\subsection{Triangular System}
\begin{assumptionS}\label{A:struct.m.l.Tri}
	Maintain Assumption \ref{A:struct.m.l}, but replace Equation (\ref{eq:status.quo.Y}) for the status quo outcome $Y$ with the following system:
	\begin{align}
		\begin{split}
			Y&= m(D,X,\varepsilon), \\
			D&= h(Z,X,\eta),
		\end{split}  \label{eq:status.quo.tri}
	\end{align}
	where $Z$ is an observed variable taking values in $\R$, and $\eta$ denotes the unobserved disturbances in the selection equation. 
\end{assumptionS}

We establish identification using a control variable approach, following an analogous construction to that in Theorem 1 of \citet{imbensNewey2009}.
\begin{lemmaC}[Control Variable]\label{lem:CF}
	In the triangular model defined by equations (\ref{eq:status.quo.tri}), assume that (a) \((\varepsilon,\eta) \indep Z|X\); (b) \(\eta\) is one-dimensional and \(p \mapsto h(Z,X,p)\) is strictly monotonic in \(p\) with probability 1; (c) for each \(x\), \(F_{\eta|X}(\cdot|x)\) is continuous and strictly increasing on its support. Then, there exists a control variable \(V=F_{D|Z,X}(D|Z,X)=F_{\eta|X}(\eta|X)\) such that $\varepsilon \indep D \,|V, X$. 
\end{lemmaC}

\begin{proposition}\label{pron:id.theta.tri}
	Under the assumptions of Theorem \ref{thm:struct.theta}, Lemma \ref{lem:CF}, and Assumption R3 (stated in Supplementary Appendix A),  we obtain
	\begin{align*}
		\theta_\Gamma = \Gamma'_{F_Y} \Big(\EEc{\big}{\omega^f(\cdot,D)\,\beta^{\mathrm{CV}}(D,V,X,\cdot)}{Y=\cdot} \Big)
	\end{align*}
	for every Hadamard-differentiable functional \(\Gamma\), where $\beta^{\mathrm{CV}}(d,v,x,y):=-\frac{\partial_d{F_{Y|D,V,X}(y|d,v,x)}}{f_{Y|D,V,X}(y|d,v,x)}$.
\end{proposition}
Proposition \ref{pron:id.theta.tri} presents a general identification formula for \(\theta_\Gamma\) in the triangular system (\ref{eq:status.quo.tri}), where \(\E[\beta^{\mathrm{CV}}(d,V,x,y)|D=d,X=x,Y=y] \) identifies \(\theta_{\mathrm{LASD}}(d,x,y)\) under $\varepsilon \indep D \,|V, X$. 

Let $\beta^{\mathrm{CQD}}(\alpha,d,v,x):=\partial_d Q_{Y|D,V,X}(\alpha|d,v,x)$. The following result extends Corollary \ref{cory:theta.Q.exo} to the triangular system.
\begin{corollary}\label{cory:theta.Q.tri}
	For \(\Gamma=Q_\tau\), under the assumptions of Proposition \ref{pron:id.theta.tri}  and assuming the support \(\supp{Y}\) is compact, we obtain 
	\begin{align}
		\theta_Q(\tau)
		&= \frac{-1}{f_Y(q_\tau)}\int \dot{\pi}(d)\,\partial_d F_{Y|D,V,X}(q_\tau|d,v,x)\,dF_{D,V,X}(d,v,x) \label{eq:Q.mms.tri}\\
		&= \int  \dot{\pi}(d) \Big( \omega_{q_\tau}(d,v,x)\, \beta^{\mathrm{CQD}}\big(\zeta_\tau(d,v,x),d,v,x\big)\Big)\,dF_{D,V,X}(d,v,x) \label{eq:Q.ffl.pron2.tri}\\
		&= \int \dot{\pi}(d)\, \beta^{\mathrm{CQD}}\big(\zeta_\tau(d,v,x),d,v,x\big)\,dF_{D,V,X|Y}(d,v,x|q_\tau) \label{eq:Q.alejo.tri}
	\end{align}
	for every $\tau\in(0,1)$, 
	where $\omega_{y}(d,v,x):=\frac{f_{Y|D,V,X}(y|d,v,x)}{f_Y(y)}$, and $\zeta_\tau(d,v,x):=\{\alpha\in(0,1):Q_{Y|D,V,X}(\alpha|d,v,x)=q_\tau\}$.
\end{corollary}

The following result extends \citeauthor{rothe2012}'s (\citeyear{rothe2012}) Theorem 4 to the triangular system, with part (ii) aligning with Theorem 1 of \citet{rothe2010el}. 
\begin{corollary}\label{cory:rothe2012.tri}
	Let the assumptions of Proposition \ref{pron:id.theta.tri} and Conditions (a)--(b) in Corollary \ref{cory:rothe2012} hold. For \(D^t=H_t^{-1}(F_D(D))\), suppose the copula \(C(\cdot,\cdot,\cdot)\) of \((D^t,V,X)\) is partially differentiable with respect to its first component. Then:
	\begin{itemize}
		\item[(i)] For a marginal perturbation \(H_t(d) = F_D(d) + t\,(G_D(d) - F_D(d))\), where \(G_D\) denotes a CDF of \(D\), we have
		\begin{align*}
			\theta_\Gamma 
			&= \int \Gamma'_{F_Y} \big(F_{Y|D,V,X}(\cdot|d,v,x)\big)\,d\left(F_{V,X|D}(v,x|d)(G_D(d)-F_D(d))\right)  \\
			&= \int \Gamma'_{F_Y}\big(F_{Y|D,V,X}(\cdot|d,v,x)\big)\,d\left(\partial_t C\big(H_t(d),F_V(v),F_X(x)\big)|_{t=0}\right).
		\end{align*}
		
		\item[(ii)] For a location shift \(H_t(d)=F_{D}(d-t)\), we have \(\theta_\Gamma=\Gamma'_{F_Y}\big(\E[\partial_d F_{Y|D,V,X}(\cdot|D,V,X)]\big)\).
	\end{itemize} 
\end{corollary}
Notably, \citet{baltagi2019} analyze counterfactual changes in the conditional distribution \(F_{D|V}\) within a triangular system, rather than changes in the marginal distribution \(F_D\). Consequently, their results do not follow from Proposition \ref{pron:id.theta.tri} in our framework.

We conclude this section by presenting in Table \ref{tb:theta.id} a summary of our identification results.
\begin{table}[htb]
	\centering
	\fontsize{9.5}{18}\selectfont
	\begin{threeparttable}
		\caption{Identification of \(\theta_\Gamma\) derived from the LASD method}\label{tb:theta.id}
		\begin{tabular}{l l l l}
			\hline 
			\multicolumn{4}{l}{Panel A. Single-Equation Model}  \\ 
			\multicolumn{4}{l}{\(\theta_{\mathrm{LASD}}(d,x,y) = \beta(d,x,y)\)  under \(\varepsilon \indep D|X\)}  \\ 
			Intervention & Functional  & Identification & Related Literature \\
			\hline 
			\(\pi_t(D)\) & \(\Gamma\)  & Proposition \ref{pron:id.theta.exo} & ------ \\
			\(\pi_t(D)\) & \(Q_\tau\)  & Corollary \ref{cory:theta.Q.exo}: Eq.(\ref{eq:Q.mms}) &  \citet{mms2024}\\
			\(D+t\) & \(Q_\tau\) & Corollary \ref{cory:theta.Q.exo}: Eqs.(\ref{eq:Q.ffl.pron2}) (\ref{eq:Q.alejo}) &  \citet{firpo2009,alejo2024}\\
			\(H^{-1}_t(F_D(D))\) & \(\Gamma\) & Corollary \ref{cory:rothe2012} &  \citet{rothe2012}\\
			\hline 
			\multicolumn{4}{l}{Panel B. Triangular System} \\
			\multicolumn{4}{l}{\(\theta_{\mathrm{LASD}}(d,x,y)=\E[\beta^{\mathrm{CV}}(d,V,x,y)|D=d,X=x,Y=y]\) with the control function \(V=F_{D|Z,X}(D|Z,X)\)} \\
			\hline 
			\(\pi_t(D)\) & \(\Gamma\)  & Proposition \ref{pron:id.theta.tri} & ------ \\
			\(\pi_t(D)\) & \(Q_\tau\)  & Corollary \ref{cory:theta.Q.tri} & ------ \\
			\(D+t\) & \(\Gamma\)  & Corollary \ref{cory:rothe2012.tri}(ii) & \citet{rothe2010el}\\
			\(H^{-1}_t(F_D(D))\) & \(\Gamma\) & Corollary \ref{cory:rothe2012.tri}(i) & ------ \\
			\hline 
		\end{tabular}	
		\begin{tablenotes}\footnotesize
			\item \textit{Notes:} The identification results for \(\theta_{\mathrm{LASD}}\) are presented in the upper portion of each panel. Then, \(\theta_\Gamma\) is identified via \(\theta_\Gamma
			= \Gamma'_{F_Y} \big( \E[\omega^f(\cdot,D)\,\theta_{\mathrm{LASD}}(D,X,\cdot)|Y=\cdot]\big)\), which follows from Theorem \ref{thm:struct.theta}.
		\end{tablenotes}
	\end{threeparttable}
\end{table}

\section{Overview of Estimation Approaches} \label{sec:discuss.esti}
The general identification results in Section \ref{sec:apply} motivate the following estimation procedure for  \(\theta_\Gamma\):
\begin{itemize}
	\item[1.]
	For specified \(\pi_t\) and \(\Gamma\), compute the Hadamard derivative \(\Gamma'_{F_Y}\) and simplify the identification formula according to Proposition \ref{pron:id.theta.exo} or \ref{pron:id.theta.tri}.
	
	\item[2.]
	Implement the plug-in estimator based on the simplified formula. For high-dimensional covariates \(X\), employ the \textit{Neyman orthogonal score} approach instead.
\end{itemize} 

As an illustration, Corollary \ref{cory:theta.Q.exo} yields the following nonparametric estimators for the quantile MPE:
\begin{align}
	\hat{\theta}_Q(\tau) 
	&= -\frac{\frac{1}{n}\sum_{i=1}^{n}\dot{\pi}(D_i)\widehat{\partial_dF}_{Y|D,X}(\hat{q}_\tau|D_i,X_i)}{\hat{f}_Y(\hat{q}_\tau)} \label{eq:esti.theta.Q.np}\\
	&=\frac{\sum_{i=1}^{n}\dot{\pi}(D_i)\,\hat{\beta}^{\mathrm{CQD}}(\hat{\zeta}_\tau(D_i,X_i),D_i,X_i)K\left(\frac{Y_i-\hat{q}_\tau}{h}\right)}{\sum_{i=1}^{n}K\left(\frac{Y_i-\hat{q}_\tau}{h}\right)}, \label{eq:esti.theta.Q.alejo}
\end{align}
where \(K(\cdot)\) denotes a kernel function. When \(\pi_t(d)=d+t\), Equation (\ref{eq:esti.theta.Q.np}) corresponds to the RIF-NP estimator of \citet{firpo2009}, and Equation (\ref{eq:esti.theta.Q.alejo}) is the reweighting estimator proposed by \citet{alejo2024}. See their respective papers for detailed computational procedures of each estimator.

When the covariate vector \(X\) is high-dimensional, the nonparametric first steps (e.g., \(\partial_dF_{Y|D,X}\)) typically lead to large biases in the plug-in estimator \(\hat{\theta}_\Gamma\) (\citealp{ddml2018,localRobust2022}). 
Using the orthogonal scores helps reduce the sensitivity with respect to the first steps to estimate \(\theta_\Gamma\).
We continue with the example of \(\theta_{Q}(\tau)\).
Let \(w:=(y,d,x)\) and \(\gamma_0(d,x,\tau):=F_{Y|D,X}(q_\tau|d,x)\). The orthogonal score of \(\theta_Q\) is given by
\begin{align*}
	\psi_{Q}(w,\theta,\gamma,\tau) := -\frac{\dot{\pi}(d)\partial_d\gamma(d,x,\tau)}{f_Y(q_\tau)} -\theta + \alpha(d,x)\frac{\I{y\leq q_\tau}-\gamma(d,x,\tau)}{f_Y(q_\tau)},
\end{align*}
where \(\alpha(d,x):= \frac{\partial_d[\dot{\pi}(d)f_{D,X}(d,x)]}{f_{D,X}(d,x)}\) denotes the Riesz representer. The function $\psi_{Q}$ satisfies the moment condition: \(\E[\psi_{Q}(W,\theta_Q(\tau),\gamma_0,\tau)]=0\), and the Neyman orthogonality condition:
\( \partial_\delta\E[\psi_{Q}(W,\theta_Q(\tau),\gamma+\delta(\gamma-\gamma_0),\tau)]|_{\delta=0}=0\)   
for every \(\gamma\) in a convex subset of some normed vector space. Then, the debiased estimator of \(\theta_Q(\tau)\) is given by 
\begin{align*}
	\check{\theta}_Q(\tau)
	&= -\frac{\frac{1}{n}\sum_{i=1}^{n}\left\{ \dot{\pi}(D_i)\partial_d\hat{F}_{Y|D,X}(\hat{q}_\tau|D_i,X_i) - \hat{\alpha}(D_i,X_i)\big(\I{Y_i \leq \hat{q}_\tau} - \hat{F}_{Y|D,X}(\hat{q}_\tau|D_i,X_i)\big) \right\}}{\hat{f}_Y(\hat{q}_\tau)},
\end{align*} 
where \(\hat{\alpha}\) can be computed using the automatic estimation approach (\citealp{adml2022}), and \(\hat{F}_{Y|D,X}\) can be estimated via the post-Lasso penalized logistic regression (\citealp{sasaki2022}). 
When \(\pi_t(d)=d+t\), \(\check{\theta}_Q(\tau)\) corresponds to the high dimensional UQR estimator proposed by \citet{sasaki2022}.

In summary, our identification framework delivers plug-in estimators for general MPEs \(\theta_\Gamma\) in low-dimensional settings. In high-dimensional settings, constructing debiased estimators requires \(\Gamma\)-specific Neyman orthogonal scores. To our knowledge, developed orthogonal score functions are lacking even for commonly used inequality-related MPEs,\footnote{In the binary treatment setting, \citet{firpo2016} develop semiparametrically efficient estimators for inequality treatment effects.} such as the Gini coefficient and the Theil index.

\section{Conclusion} \label{sec:conclusion}
This paper demonstrates that, under regularity conditions, the marginal policy effect (MPE) for a Hadamard-differentiable functional equals its functional derivative evaluated at the outcome-conditioned weighted average structural derivative. This equivalence is definitional rather than dependent on specific identification assumptions. The main theorem offers an alternative identification strategy for the MPE, using the identification of the Local Average Structural Derivative (LASD) as an intermediate step.

This study points to several avenues for future research. First, the LASD identification framework can be extended to other nonseparable structural models, such as panel data models, where identifying the relevant LASD objects poses the central theoretical challenge. Second, developing debiased estimators for commonly used inequality-related MPEs (e.g., the Gini coefficient) in high-dimensional settings remains an important avenue for future work.



\begin{titlepage}
	\onehalfspacing
	\hypersetup{linkcolor=Black}
	
	\begin{center}
		\Large
		Supplementary Appendix to ``\maintitle''
		
		\vspace{1.5em}
		\normalsize\scshape
		Zhixin Wang, Yu Zhang, and Zhengyu Zhang
	\end{center}
	
	\renewcommand{\abstractname}{}
	\begin{abstract}
		In this supplementary appendix, Section \ref{sec:pf.main} provides the proofs of the main results, while Section \ref{sec:discussion} presents additional results on the average and Gini marginal policy effects that were omitted from the main text.
	\end{abstract}

	\vspace{1cm}
	\ToC
	\thispagestyle{empty}
	
\end{titlepage}

\renewcommand{\thepage}{S.\arabic{page}}
\begin{appendices}

	\allowdisplaybreaks
	
	\section{Proof of Main Results}\label{sec:pf.main}
	\renewcommand{\theequation}{A.\arabic{equation}}
	\setcounter{equation}{0}
	Analogous to Assumption R2(a)--(c), we introduce the following conditions for identification using the control variable \(V\). 
	\setcounter{assumptionR}{2}
	\begin{assumptionR}[Regularity] \label{A:regular.hm.cf}
		For each $(d,v,x)\in\supp{D}\times  \supp{V} \times \supp{X}$:
		\begin{itemize}
			\item[(a)] The conditional CDF of \(Y\) given \((D,V,X)\) is absolutely continuous, and its density $f_{Y|D,V,X}(\cdot|d,v,x)$ is strictly positive, bounded and continuous on its support. Moreover, $ F_{Y|D,V,X}(y|d,v,x)$ is continuously differentiable in \(d\) for each $(y,v,x)$, with derivative denoted by \(\partial_dF_{Y|D,V,X}(\cdot|\cdot,\cdot,\cdot)\).
			
			\item[(b)] $\partial_d m(d,x,\cdot)$ is measurable and satisfies
			\[
			\pbc{\Big}{\big| m(d+\delta,x,\varepsilon) - m(d,x,\varepsilon) - \delta\,\partial_d m(d,x,\varepsilon) \big| \geq \nu\delta}{D=d,V=v,X=x} = o(\delta)
			\]
			as $\delta \downarrow 0$ for every $\nu > 0$.
			\item[(c)] The conditional CDF of \((Y,\partial_d m(d,x,\varepsilon))\) given \((D,V,X)\) is absolutely continuous, and its density $f_{Y,\partial_dm^{d,x}|D,V,X}(\cdot,y'|d,v,x)$ is continuous for each $(y',d,v,x)$. Moreover, there exists a Lebesgue integrable function $g:\R\to\R$ satisfying $\int |y'g(y')|\,dy'<\infty$, such that  $f_{Y,\partial_dm^{d,x}|D,V,X}(y,y'|d,v,x) \leq C |g(y')|$ for some constant $C$.
		\end{itemize}
	\end{assumptionR}
	\subsection{Proof of Theorem \ref{thm:struct.theta}}
	\begin{proof}
		Under Assumptions S1 and R1, we have
		\begin{align}
			F_{Y^t}(y) - F_Y(y)
			&=\pbb{\big}{m(\pi_t(D),X,\varepsilon)\leq y} - \pbb{\big}{Y \leq y} \notag\\
			&=\pbb{\big}{Y\leq y - \{m(\pi_t(D),X,\varepsilon)-m(D,X,\varepsilon)\}} - \pbb{\big}{Y \leq y} \notag\\
			&=\pbb{\big}{y < Y \leq y - \{m(\pi_t(D),X,\varepsilon)-m(D,X,\varepsilon)\}} \notag\\
			&\quad - \pbb{\big}{y - \{m(\pi_t(D),X,\varepsilon)-m(D,X,\varepsilon)\} < Y \leq y} \notag\\
			&\eq{1}{=}\pbb{\big}{y \leq Y \leq y - t\,\dot{m}(D,X,\varepsilon)} \notag\\
			&\quad - \pbb{\big}{y - t\,\dot{m}(D,X,\varepsilon) \leq Y \leq y} + o(t) \notag\\
			&=\pbb{\bigg}{Y \geq y,\, \dot{m}(D,X,\varepsilon) \leq \frac{y-Y}{t}} \notag\\
			&\quad -\pbb{\bigg}{Y \leq y,\, \dot{m}(D,X,\varepsilon) \geq \frac{y-Y}{t}} + o(t) \notag\\
			&= \int_{y}^{\infty}\int_{-\infty}^{\frac{y-a}{t}} f_{Y,\dot{m}}(a,y') \,dy'da \notag\\
			&\quad - \int_{-\infty}^{y}\int_{\frac{y-a}{t}}^{\infty}f_{Y,\dot{m}}(a,y') \,dy'da  + o(t) \notag\\
			&\eq{2}{=} t \int_{-\infty}^{0}\int_{-\infty}^{u}f_{Y,\dot{m}}(y-tu,y') \,dy'du \notag\\
			&\quad - t \int_{0}^{\infty}\int_{u}^{\infty}f_{Y,\dot{m}}(y-tu,y') \,dy'\,du + o(t) \notag\\
			&\eq{3}{=} t \int_{-\infty}^{0}\int_{y'}^{0} f_{Y,\dot{m}}(y-tu,y') \,du\,dy' \notag\\
			&\quad - t \int_{0}^{\infty}\int_{0}^{y'} f_{Y,\dot{m}}(y-tu,y') \,du\,dy' + o(t) \notag\\
			&=t \int_{-\infty}^{\infty}\int_{y'}^{0} f_{Y,\dot{m}}(y-tu,y') \,du\,dy' + o(t) \notag
		\end{align}
		as \(t\downarrow 0\), where equality (1) follows from Assumptions R1(b); equality (2) follows from the change of variable \(u=(y-a)/t\); equality (3) follows from interchanging the order of integration.  
		Then, dividing both sides by \(t\) and taking the limit \(t\downarrow 0\) implies that
		\begin{align*}
			\theta_{id}(y) 
			&= \lim_{t\downarrow 0} \int_{-\infty}^{\infty}\int_{y'}^{0} f_{Y,\dot{m}}(y-tu,y') \,du\,dy' \\
			&\eq{4}{=} \int_{-\infty}^{\infty}\int_{y'}^{0} \lim_{t \downarrow 0} f_{Y,\dot{m}}(y-tu,y') \,du\,dy' \\
			&= -\int_{-\infty}^{\infty} y'\,f_{Y,\dot{m}}(y,y') \,dy' \\
			&\eq{5}{=} -\EEc{\big}{f_Y(y)\dot{m}(D,X,\varepsilon)}{Y=y} \\
			&= \EEc{\big}{\omega^f(y,D)\partial_d m(D,X,\varepsilon)}{Y=y}
		\end{align*}
		for each \(y \in \supp{Y}\), where equality (4) follows by using the dominance convergence theorem (DCT) under Assumption R1(c); equality (5) follows from Assumption R1(a). This proves the first part.
		
		For the second part, since \(\Gamma(F_Y)\) is Hadamard differentiable at \(F_Y\) with derivative \(\Gamma'_{F_Y}\), we obtain that
		\begin{align*}
			\theta_\Gamma 
			&= \lim_{t \downarrow 0} \frac{\Gamma(F_{Y} + t\,h_t)-\Gamma(F_Y)}{t} \\
			&= \Gamma'_{F_Y}\big(\theta_{id}\big)\\
			&= \Gamma'_{F_Y}\big(\EEc{\big}{\omega^f(\cdot,D)\partial_d m(D,X,\varepsilon)}{Y=\cdot}\big),
		\end{align*}
		where \(h_t:=\frac{F_{Y^t}-F_Y}{t} \to \theta_{id}\) as \(t \downarrow 0\). This completes the proof.
	\end{proof}

	\subsection{Proof of Theorem \ref{thm:struct.uqr.reg}}
	\begin{proof}
		Let \(q_\alpha(d,x):=Q_{Y|D,X}(\alpha|d,x)\). Under Assumption R2(a),
		differentiating both sides of the equality \(\alpha \equiv F_{Y|D,X}(q_\alpha(d,x)|d,x) \) with respect to \(d\) yields
		\begin{align*}
			\partial_d q_\alpha(d,x) = -\frac{\partial_dF_{Y|D,X}(q_\alpha(d,x)|d,x)}{f_{Y|D,X}(q_\alpha(d,x)|d,x)}.
		\end{align*}
		Then, applying Theorem 4.1 of \citet{hm2009}, we obtain 
		\begin{align*}
			&-\partial_d F_{Y|D,X}(q_\alpha(d,x)|d,x)\\
			&\; = f_{Y|D,X}(q_\alpha(d,x)|d,x)  \EEc{\big}{\partial_dm(d,x,\varepsilon)}{D=d,X=x,Y=q_\alpha(d,x)}\\
			&\quad  - \EEc{\big}{\I{m(d,x,\varepsilon)\leq q_\alpha(d,x)}\,\partial_d\ln\big(f_{\varepsilon|D,X}(\varepsilon|d,x)\big)}{D=d,X=x}.
		\end{align*}
		
		Let \(q_\tau :=Q_Y(\tau)\). Note that \(f_{Y|D,X}(\cdot|d,x) > 0\) in Assumption R2(a) implies that \(\zeta_\tau(d,x):=\{\alpha\in(0,1):q_\alpha(d,x)=q_\tau\}\) is a singleton. Hence, setting \(\alpha = \zeta_\tau(d,x)\) and integrating with respect to \(F_{D,X}\) implies 
		\begin{align*}
			-\EE{\big}{\partial_d F_{Y|D,X}(q_\tau|D,X)}
			&= \int f_{Y|D,X}(q_\tau|d,x)\int  \partial_dm(d,x,e)\,f_{\varepsilon|D,X,Y}(e|d,x,q_\tau)\,de\,dF_{D,X}(d,x)  \\
			&\quad - \EE{\big}{\I{m(D,X,\varepsilon)\leq q_\tau}\,\partial_d\ln\big(f_{\varepsilon|D,X}(\varepsilon|D,X)\big)}\\
			&=\int \partial_dm(d,x,e)\,dF_{D,X,\varepsilon|Y}(d,x,e|q_\tau) \cdot f_Y(q_\tau) \\
			&\quad - \EE{\big}{\I{m(D,X,\varepsilon)\leq q_\tau}\,\partial_d\ln\big(f_{\varepsilon|D,X}(\varepsilon|D,X)\big)}.
		\end{align*}
		Under Assumption R1(a), dividing both sides by \(f_Y(q_\tau)\) completes the proof.
	\end{proof}
	\subsection{Proof of Corollary \ref{cory:struct.theta.Q}}
	\begin{proof}
		Under the assumptions given in the corollary, applying Lemma 21.4 of \citet{vaart2000} implies that $Q_\tau(F_Y)$ is Hadamard differentiable at $F_Y$, with derivative
		\begin{align}
			[Q_\tau]'_{F_Y}(h) = - \bigg(\frac{h}{f_Y}\circ Q_Y\bigg)(\tau) = -\frac{h(q_\tau)}{f_Y(q_\tau)}, \quad q_\tau:=Q_Y(\tau). \label{eq:Q.hadamard.dev}
		\end{align}
		where \(h\in\ell^\infty(\supp{Y})\). Therefore, \(\omega^f(q_\tau,d) = -f_Y(q_\tau)\,\dot{\pi}(d)\) and 
		\begin{align*}
			\theta_Q(\tau) 
			&= -\frac{ \EEc{\big}{\omega^f(q_\tau,D)\, \partial_d m(D,X,\varepsilon)}{Y=q_\tau} }{f_Y(q_\tau)}\\
			&= \EEc{\big}{\dot{\pi}(D)\,\partial_d m(D,X,\varepsilon)}{Y=q_\tau}.
		\end{align*}
		This completes the proof.
	\end{proof}
	
	\subsection{Proof of Proposition \ref{pron:id.theta.exo}}
	\begin{proof}
		First, under the conditional independence \(\varepsilon \indep D |X\), applying Theorem 2 or \citet[Theorem 2.1]{hm2007} gives
		\begin{align*}
			\theta_{\mathrm{LASD}}(d,x,q_\alpha(d,x))
			= \beta^{\mathrm{CQD}}(\alpha,d,x)
			= 
			-\frac{\partial_d F_{Y|D,X}(q_\alpha(d,x)|d,x) }{ f_{Y|D,X}(q_\alpha(d,x)|d,x)}.
		\end{align*}  
		Setting \(\alpha=F_{Y|D,X}(y|d,x)\) yields \(q_\alpha(d,x)=y\), and thus
		\[\theta_{\mathrm{LASD}}(d,x,y) = -\frac{\partial_d F_{Y|D,X}(y|d,x) }{f_{Y|D,X}(y|d,x)}=:\beta(d,x,y)\]
		for each \(y \in \supp{Y|D,X}\).  Therefore, applying Theorem 1, we obtain that
		\begin{align*}
			\theta_\Gamma 
			&= \Gamma'_{F_Y}\big(\EEc{\big}{\omega^f(\cdot,D)\,\partial_d m(D,X,\varepsilon)}{Y=\cdot}\big) \\
			&= \Gamma'_{F_Y}\big(\EEc{\big}{\omega^f(\cdot,D)\,\underbrace{\E[\partial_dm(D,X,\varepsilon)|D,X,Y=\cdot]}_{=: \theta_{\mathrm{LASD}}(D,X,\cdot)}}{Y=\cdot}\big)\\
			&= \Gamma'_{F_Y}\big(\EEc{\big}{\omega^f(\cdot,D)\, \beta(D,X,\cdot)}{Y=\cdot}\big) .
		\end{align*}
		This completes the proof.
	\end{proof}
	
	\subsection{Proof of Corollary \ref{cory:theta.Q.exo}}
	\begin{proof}
		The Hadamard derivative of \(\Gamma=Q_\tau\) is given by Equation (\ref{eq:Q.hadamard.dev}). Then, applying Proposition 1 completes the proof.
	\end{proof}
	
	\subsection{Proof of Corollary \ref{cory:rothe2012}}
	\begin{proof}
		\textit{(i)}
		Let \(U:=F_D(D)\) denote the rank of \(D\). Setting the policy function as \(\pi_t(d)=H_t^{-1}(F_D(d))\), we have \(D^t:=\pi_t(D)=H_t^{-1}(U)\), and under the condition that \(H_t\) is continuous, \(H_t(D^t)=H_t(H_t^{-1}(U))=U\).  Next, for a given direction CDF \(G_D\), let \(H_t(d)=F_D(d) + t\,(G_D(d)-F_D(d))\) denote a marginal perturbation. Under the condition \(\supp{D_t} \subseteq \supp{D}\), setting \(d=D^t\) yields
		\begin{align*}
			U = F_{D}(\pi_t(D)) + t\cdot(G_D(\pi_t(D)) - F_D(\pi_t(D))). 
		\end{align*} 
		Differentiating both sides with respect to \(t\) at \(t=0\) implies that
		\begin{align*}
			0 &= f_D(D)\dot{\pi}(D) + G_D(D) - F_D(D) \tag{\(\pi_0(d)\equiv d\)} \\
			\implies 
			\dot{\pi}(D) &= -\frac{G_D(D)-F_D(D)}{f_D(D)} \tag{\(f_D > 0\)}, 
		\end{align*}
		where \(\dot{\pi}(d):=\partial_t\pi_t(d)|_{t=0}\).  
		Applying Proposition 1, we have 
		\begin{align*}
			\theta_\Gamma 
			&= \Gamma'_{F_Y}\bigg( \iint f_Y(\cdot)\,\frac{G_D(d)-F_D(d)}{f_D(d)}\,\bigg(-\frac{\partial_d F_{Y|D,X}(\cdot|d,x) }{f_{Y|D,X}(\cdot|d,x)}\bigg) f_{D,X|Y}(d,x|\cdot)\,dd\,dx \bigg) \\
			&=\Gamma'_{F_Y}\bigg( -\iint f_{X|D}(x|d) (G_D(d)-F_D(d))\,\partial_d F_{Y|D,X}(\cdot|d,x) \,dd\,dx \bigg)\\
			&=\Gamma'_{F_Y}\bigg( \iint F_{Y|D,X}(\cdot|d,x) \,d_d [f_{X|D}(x|d) (G_D(d)-F_D(d))]\,dx \bigg)\\
			&= \Gamma'_{F_Y}\bigg( \int F_{Y|D,X}(\cdot|d,x) \,d\big(F_{X|D}(x|d) (G_D(d)-F_D(d))\big) \bigg) \\
			&=  \int \Gamma'_{F_Y}\big(F_{Y|D,X}(\cdot|d,x)\big) \,d\big(F_{X|D}(x|d) (G_D(d)-F_D(d))\big)
		\end{align*}
		where the third equality follows from integration by parts; note that both \(G_D\) and \(F_D\) are CDFs, which implies \(G_D(\cdot)-F_D(\cdot)=0\) on the boundary of the support. The last equality follows from the linearity of the Hadamard derivative \(\Gamma'_{F_Y}\). 
		
		We next derive Theorem 4 of \citet{rothe2012} to verify our claim.
		Let \(C:[0,1]^{1+d_x}\to[0,1]\) denote the copula. 
		\citet{rothe2012} provides
		\begin{align*}
			\theta_\Gamma 
			&= \int \Gamma'_{F_Y}\big(F_{Y|D,X}(\cdot|d,x)\big)\,d\big(\partial_1C(F_D(d),F_X(x))(G_D(d)-F_D(d))\big)\\
			&=\int \Gamma'_{F_Y}\big(F_{Y|D,X}(\cdot|d,x)\big)\,d\big(\partial_tC(H_t(d),F_X(x))|_{t=0}\big).
		\end{align*}
		
		Note that \(H_t\) is the CDF of counterfactual policy \(D^t\), by Sklar's theorem (see, e.g., \citet{fanCopula2014}), there exist a copula \(C:[0,1]^{1+d_\mathsf{x}}\to[0,1]\) such that
		\[ C\big(H_t(d),F_X(x)\big) = F_{D^t,X}(d,x) ,\]
		for every \((d,x) \in \supp{D^t} \times \supp{X}\). Differentiating both sides with respect to \(t\) at \(t=0\) implies that
		\begin{align*}
			&\partial_t C\big(H_t(d),F_X(x)\big)|_{t=0} = \partial_t F_{D^t,X}(d,x) |_{t=0} \\
			&\;= \partial_t \pbb{\big}{H_t^{-1}(F_D(D))\leq d,\, X \leq x}|_{t=0} \tag{\(\pi_t(\cdot)=H_t^{-1}(F_D(\cdot))\)}\\
			&\;= \partial_t \pbb{\big}{F_D(D)\leq H_t(d),\, X \leq x}|_{t=0} \tag{\(H_t\) is continuous}\\
			&\;= \partial_t \pbb{\big}{D\leq Q_D(H_t(d)),\, X \leq x}|_{t=0}  \tag{\(f_D > 0\)}\\
			&\;= \partial_t F_{D,X}(Q_D(H_t(d)),x)\\
			&\;= \partial_d F_{D,X}(d,x) \cdot \partial_\tau Q_D(\tau)|_{\tau=F_{D}(d)} \cdot \partial_tH_t(d)|_{t=0} \tag{\(F_{D^0} \equiv F_D\)}\\
			&\,\eq{1}{=} \partial_d\bigg\{\int^{d}\!\!\int^{x} f_{D,X}(d',x')\,dx'\,dd'\bigg\} \cdot \frac{1}{f_D(Q_D(\tau))}\Big|_{\tau=F_D(d)}\cdot\big(G_D(d)-F_D(d)\big) \\
			&\;=\int^{x} f_{D,X}(d,x')\,dx'\cdot \frac{ G_D(d)-F_D(d)}{f_D(d)} \\
			&\;=F_{X|D}(x|d)\,\big(G_D(d)-F_D(d)\big),
		\end{align*}
		where equality (1) follows from differentiating both sides with respect to \(\tau\) of identity \(\tau \equiv F_D(Q_D(\tau))\). This proves part (i) of the corollary.
		
		\textit{(ii)} The location shift \(H_t(d)=F_D(d-t)=\pb[D\leq d-t]\) is equivalent to the counterfactual change \(\pi_t(D)=D+t\), and thus \(\dot{\pi}(d)\equiv1\). Then, applying Proposition 1 directly yields that
		\begin{align*}
			\theta_\Gamma
			&=\Gamma'_{F_Y}\bigg( \iint f_Y(\cdot)\,\frac{\partial_d F_{Y|D,X}(\cdot|d,x) }{f_{Y|D,X}(\cdot|d,x)} f_{D,X|Y}(d,x|\cdot)\,dd\,dx \bigg)\\
			&=\Gamma'_{F_Y}\Big( \EE{\big}{\partial_d F_{Y|D,X}(\cdot|D,X)} \Big).
		\end{align*}
		This completes the proof.
	\end{proof}

	\subsection{Proof of Proposition \ref{pron:id.theta.tri}}
	\begin{proof}
		Identifying the \(\theta_{\mathrm{LASD}}\) using the control variable method follows the same line of reasoning as in \citet{hm2007}. Under Assumption R3(a), \(\alpha \equiv F_{Y|D,V,X}(q_\alpha(d,v,x)|d,v,x)\) for \( 0< \alpha < 1\), where \(q_\alpha(d,v,x):=Q_{Y|D,V,X}(\alpha|d,v,x)\). Then, we have  
		\begin{align}
			0 
			&= \pbc{\big}{Y\leq q_\alpha(d+\delta,v,x)}{D=d+\delta,V=v,X=x} \notag\\
			&\quad - \pbc{\big}{Y\leq q_\alpha(d,v,x)}{D=d,V=v,X=x} \notag\\
			&= \pbc{\big}{m(d+\delta,x,\varepsilon)\leq q_\alpha(d+\delta,v,x)}{D=d+\delta,V=v,X=x} \notag\\
			&\quad-\pbc{\big}{m(d,x,\varepsilon)\leq q_\alpha(d,v,x)}{D=d,V=v,X=x}\notag\\
			&= A_{1,\delta}(d,v,x) + A_{2,\delta}(d,v,x) +A_{3,\delta}(d,v,x), \label{eq:hm.decomp.cv}
		\end{align}
		for \(\delta > 0\) close to zero, where
		\begin{align*}
			\begin{split}
				A_{1,\delta}(d,v,x) 
				&=  \pbc{\big}{m(d+\delta,x,\varepsilon)\leq q_\alpha(d+\delta,v,x)}{D=d+\delta,V=v,X=x}\\
				&\quad - \pbc{\big}{m(d+\delta,x,\varepsilon)\leq q_\alpha(d,v,x)}{D=d+\delta,V=v,X=x}\\
				A_{2,\delta}(d,v,x)
				&=\pbc{\big}{m(d+\delta,x,\varepsilon)\leq q_\alpha(d,v,x)}{D=d+\delta,V=v,X=x}\\
				&\quad - \pbc{\big}{m(d+\delta,x,\varepsilon)\leq q_\alpha(d,v,x)}{D=d,V=v,X=x}\\
				A_{3,\delta}(d,v,x)
				&=\pbc{\big}{m(d+\delta,x,\varepsilon)\leq q_\alpha(d,v,x)}{D=d,V=v,X=x}\\
				&\quad - \pbc{\big}{m(d,x,\varepsilon)\leq q_\alpha(d,v,x)}{D=d,V=v,X=x}.
			\end{split} 
		\end{align*}
		By Lemma C, there exists \(V=F_{D|Z,X}(D|Z,X)\) such that \(\varepsilon \indep D |X, V\), and thus \(A_{2,\delta}(d,v,x) = 0\). Next, under Assumption R3, following the similar procedure to the proof of \citet[Theorem 2.1]{hm2007}, we obtain
		\begin{align*}
			A_{1,\delta}(d,v,x) 
			&= \delta \partial_dq_\alpha(d,v,x)f_{Y|D,V,X}(q_\alpha(d,v,x)|d,v,x) + o(\delta)\\
			&= -\delta \partial_d F_{Y|D,V,X}(q_\alpha(d,v,x)|d,v,x) + o(\delta)\\
			A_{3,\delta}(d,v,x) 
			&=-\delta \int_{-\infty}^{\infty}y' f_{Y,\partial_dm^{d,x}|D,V,X}(q_\alpha(d,v,x),y'|d,v,x) \,dy' + o(\delta)\\
			&=- \delta\EEc{\big}{\partial_dm(d,x,\varepsilon)}{D=d,V=v,X=x,Y=q_\alpha(d,v,x)}\\
			&\qquad \cdot f_{Y|D,V,X}(q_\alpha(d,v,x)|d,v,x) + o(\delta)
		\end{align*} 
		as \(\delta \downarrow  0\). 
		Then, dividing both sides of Equation (\ref{eq:hm.decomp.cv}) by \(\delta\) and taking the limit \(\delta \downarrow 0\) implies that
		\begin{align*}
			\beta^{\mathrm{CV}}(d,v,x,q_\alpha(d,v,x)) =  \EEc{\big}{\partial_dm(d,x,\varepsilon)}{D=d,V=v,X=x,Y=q_\alpha(d,v,x)}
		\end{align*}
		where \(\beta^{\mathrm{CV}}(d,v,x,y):=  -\frac{\partial_d F_{Y|D,V,X}(y|d,v,x)}{ f_{Y|D,V,X}(y|d,v,x)}\). By setting \(\alpha=F_{Y|D,V,X}(y|d,v,x)\), \(\theta_{\mathrm{LASD}}\) is identified as
		\begin{align*}
			\theta_{\mathrm{LASD}}(d,x,y) = \EEc{\big}{\beta^{\mathrm{CV}}(d,V,x,y)}{D=d,X=x,Y=y}
		\end{align*} 
		for each \((d,x,y) \in \supp{D}\times\supp{X}\times\supp{Y|D,V,X}\).
		Therefore, using Theorem 1, we obtain that
		\begin{align*}
			\theta_\Gamma
			&= \Gamma'_{F_Y}\big(\EEc{\big}{\omega^f(\cdot,D)\,\partial_d m(D,X,\varepsilon)}{Y=\cdot}\big) \\
			&=\Gamma'_{F_Y} \big(\EEc{\big}{\omega^f(\cdot,D)\,\theta_{\mathrm{LASD}}(D,X,\cdot)}{Y=\cdot} \big)\\
			&=\Gamma'_{F_Y} \big(\EEc{\big}{\omega^f(\cdot,D)\,\EEc{\big}{\beta^{\mathrm{CV}}(D,V,X,\cdot)}{D,X,Y=\cdot}}{Y=\cdot} \big)\\
			&= \Gamma'_{F_Y} \big(\EEc{\big}{\omega^f(\cdot,D)\,\beta^{\mathrm{CV}}(D,V,X,\cdot)}{Y=\cdot} \big).
		\end{align*}
		This completes the proof.
	\end{proof}
	
	\subsection{Proof of Corollary \ref{cory:theta.Q.tri}}
	\begin{proof}
		The Hadamard derivative of \(\Gamma=Q_\tau\) is given by Equation (\ref{eq:Q.hadamard.dev}). Then, applying Proposition 2 completes the proof.
	\end{proof}
	
	\subsection{Proof of Corollary \ref{cory:rothe2012.tri}}
	\begin{proof}
		\textit{(i)}\, In the proof of Corollary 2, we have shown that \(\dot{\pi}(D) = -(G_D(D)-F_D(D))/f_D(D)\) for  \(D^t=\pi_t(D)=H_t^{-1}(F_D(D))\). Applying Proposition 2, we have
		\begin{align*}
			\theta_\Gamma
			&= \Gamma'_{F_Y}\bigg(\int f_Y(\cdot)\,\frac{G_D(d)-F_D(d)}{f_D(d)}\,\bigg(-\frac{\partial_dF_{Y|D,V,X}(\cdot|d,v,x)}{f_{Y|D,V,X}(\cdot|d,v,x)}\bigg) \,dF_{D,V,X|Y}(d,v,x|\cdot)\bigg)\\
			&=\Gamma'_{F_Y}\bigg(\iiint f_{V,X|D}(v,x|d)\,(G_D(d)-F_D(d))\,\partial_dF_{Y|D,V,X}(\cdot|d,v,x)\,dd\,dx\,dv\bigg) \\
			&=\Gamma'_{F_Y}\bigg(\iiint F_{Y|D,V,X}(\cdot|d,v,x)\,d_d[f_{V,X|D}(v,x|d)\,(G_D(d)-F_D(d))]\,dx\,dv\bigg) \\
			&=\Gamma'_{F_Y}\bigg(\int F_{Y|D,V,X}(\cdot|d,v,x)\,d\big(F_{V,X|D}(v,x|d)\,(G_D(d)-F_D(d))\big)\bigg).
		\end{align*} 
		where the third equality follows from integration by parts. 
		Since \(H_t\) is another CDF of \(D^t\), by Sklar's theorem, there exists a copula \(C:[0,1]^{2+d_\mathsf{x}} \to [0,1]\) such that
		\[C(H_t(d),F_V(v),F_X(x)) = F_{D^t,V,X}(d,v,x)\]
		for each \((d,v,x)\in\supp{D^t} \times \supp{V} \times \supp{X}\). Differentiating both sides with respect to \(t\) at \(t=0\) implies that
		\begin{align*}
			&\partial_t C(H_t(d),F_V(v),F_X(x))|_{t=0} 
			= \partial_tF_{D^t,V,X}(d,v,x)|_{t=0}\\
			&\;= \partial_t\pbb{\big}{H_t^{-1}(F_D(D)) \leq d, V\leq v, X\leq x}|_{t=0} \tag{\(\pi_t(\cdot)=H_t^{-1}(F_D(\cdot))\)} \\
			&\;= \partial_t\pbb{\big}{D\leq Q_D(H_t(d)), V\leq v, X\leq x}|_{t=0} \tag{\(H_t\) is continuous \& \(f_D > 0\)} \\
			&\;= \partial_tF_{D,V,X}(Q_D(H_t(d)),v,x)|_{t=0}\\
			&\;= \partial_{d'}F_{D,V,X}(d',v,x)|_{d'=d} \cdot \partial_\tau Q_D(\tau)|_{\tau = F_{D}(d)}\cdot \partial_tH_t(d)|_{t=0}\\
			&=\partial_{d'}\biggl\{\int^{d'}\!\!\int^x\!\!\int^v f_{D,V,X}(a,v,x)\,dv'\,dx'\,da\biggr\}\Big|_{d'=d}\cdot \frac{G_D(d)-F_D(d)}{f_D(d)}\\
			&=\int^x\!\!\int^v f_{D,V,X}(d,v,x)\,dv'\,dx' \cdot \frac{G_D(d)-F_D(d)}{f_D(d)}\\
			&= F_{V,X|D}(v,x|d)\,\big(G_D(d) - F_D(d)\big). 
		\end{align*}
		This proves part (i) of the corollary.
		
		\textit{(ii)}\, 
		For the location shift, $\dot{\pi}(d)\equiv 1$. Applying Proposition 2, we have
		\begin{align*}
			\theta_\Gamma
			&=\Gamma'_{F_Y}\bigg(\int -f_Y(\cdot)\,\bigg(-\frac{\partial_d{F_{Y|D,V,X}(\cdot|d,v,x)}}{f_{Y|D,V,X}(\cdot|d,v,x)}\bigg)\,dF_{D,V,X|Y}(d,v,x|\cdot) \bigg) \\
			&=\Gamma'_{F_Y}\Big(\EE{\big}{\partial_d{F_{Y|D,V,X}(\cdot|D,V,X)}}\Big) ,
		\end{align*}
		which is consistent with Theorem 1 of \citet{rothe2010el}. This completes the proof.
	\end{proof}
	
	\section{Additional Results} \label{sec:discussion}
	\renewcommand{\theequation}{B.\arabic{equation}}
	\subsection{Average Policy Effects}
	For $\Gamma=\mu: F \mapsto \int y\,dF(y)$, the average MPE is defined as 
	\[\theta_\mu:=\lim_{t \downarrow 0}\frac{\mu(F_{Y^t})-\mu(F_Y)}{t} = \lim_{t\downarrow 0}\frac{\E[m(\pi_t(D),X,\varepsilon)]-\E[Y]}{t}.\]

	\subsubsection{Structural Interpretation}
	\begin{corollaryS}[Structural characterization of \(\theta_\mu\)]~\\ \label{cory:struct.theta.mu}
		For \(\Gamma=\mu\), under the assumptions of Theorem 1 and assuming the support \(\supp{Y}\) is compact, we obtain
		\begin{align*}
			\theta_\mu = \EE{\big}{\dot{\pi}(D)\,\partial_dm(D,X,\varepsilon)}.
		\end{align*}
	\end{corollaryS}
	\begin{proof}[\normalfont\bfseries Proof of Corollary \ref{cory:struct.theta.mu}]
		Under the assumptions given in the Corollary, \(\mu(F_Y) = \int y\,dF_Y(y) = \int_{0}^{1}\,Q_\tau(F_Y)\,d\tau\), and
		\(\mu\) is Hadamard differentiable at \(F_Y\) with derivative 
		\begin{align}
			\mu'_{F_Y}(h) = \int_{0}^{1} -\frac{h(q_\tau)}{f_Y(q_\tau)}\,d\tau,\quad q_\tau := Q_\tau(F). \label{eq:mu.hadamard.dev}
		\end{align}
		Then, applying Theorem 1 in the main text, we have
		\begin{align*}
			\theta_\mu 
			&= \int_{0}^{1} \EEc{\big}{\dot{\pi}(D)\partial_dm(D,X,\varepsilon)}{Y=q_\tau}\,d\tau\\
			&\eq{1}{=} \int \EEc{\big}{\dot{\pi}(D)\partial_dm(D,X,\varepsilon)}{Y=y}\,dF_Y(y)\\
			&=\EE{\big}{\dot{\pi}(D)\partial_dm(D,X,\varepsilon)},
		\end{align*}
		where equality (1) follows from the change of variable \(\tau = F_Y(y)\). This completes the proof.
	\end{proof}
	
	\subsubsection{Identification}
	\begin{corollaryS}[Identification of \(\theta_\mu\) under conditional independence]~\\ \label{cory:id.exo.theta.mu}
		For \(\Gamma=\mu\), under the assumptions of Proposition 1 and assuming the support \(\supp{Y}\) is compact, we obtain
		\begin{align}
			\theta_\mu 
			&= \EE{\big}{\dot{\pi}(D)\beta(D,X,Y)} \label{eq:id.exo.theta.mu}\\
			&= \EE{\big}{\dot{\pi}(D)\,\partial_d\E[Y|D,X]} \label{eq:id.exo.theta.mu.wooldridge}
		\end{align}
		where \(\beta(d,x,y):=-\frac{\partial_d{F_{Y|D,X}(y|d,x)}}{f_{Y|D,X}(y|d,x)}\). 
	\end{corollaryS}
	
	\begin{remark}
		Equality (\ref{eq:id.exo.theta.mu.wooldridge}) is the main result of \citet{wooldridge2004}, where the right-hand side's estimand is the $\dot{\pi}$-weighted \textit{conditional average partial derivative}.
	\end{remark}
	
	\begin{proof}[\normalfont\bfseries Proof of Corollary \ref{cory:id.exo.theta.mu}]
		Equality (\ref{eq:id.exo.theta.mu}) follows from Proposition 1. For \(\Gamma=\mu\), we have
		\begin{align*}
			\theta_\mu
			&=\int_{0}^{1} - \frac{\EEc{\big}{\omega^f(q_\tau,D)\beta(D,X,q_\tau)}{Y=q_\tau}}{f_Y(q_\tau)}\,d\tau\\
			&=\int_{0}^{1} \EEc{\big}{\dot{\pi}(D)\beta(D,X,q_\tau)}{Y=q_\tau} \,d\tau\\
			&\eq{1}{=}\int_{\supp{Y}}\EEc{\big}{\dot{\pi}(D)\beta(D,X,y)}{Y=y} \,dF_Y(y)\\
			&= \EE{\big}{\dot{\pi}(D)\beta(D,X,Y)},
		\end{align*}
		where equality (1) follows from the change of variable \(\tau = F_Y(y)\). 
		
		Equality (\ref{eq:id.exo.theta.mu.wooldridge}) follows from Corollary \ref{cory:struct.theta.mu}. Note that
		\begin{align*}
			\theta_\mu
			&=\E[\dot{\pi}(D)\,\partial_d m(D,X,\varepsilon)]\\
			&\eq{2}{=}\int \dot{\pi}(d)\EEc{\bigg}{\lim_{\delta \to 0}\frac{m(d+\delta,x,\varepsilon)-m(d,x,\varepsilon)}{\delta}}{D=d,X=x}\,dF_{D,X}(d,x)\\
			&=\int \dot{\pi}(d)\,\lim_{\delta \to 0}\frac{\E[m(d+\delta,x,\varepsilon)|D=d,X=x]-\E[Y|D=d,X=x]}{\delta}\,dF_{D,X}(d,x) \\
			&\eq{3}{=} \int \dot{\pi}(d)\,\lim_{\delta \to 0}\frac{\E[Y|D=d+\delta,X=x]-\E[Y|D=d,X=x]}{\delta}\,dF_{D,X}(d,x) \\
			&=\int \dot{\pi}(d)\,\partial_d\E[Y|D=d,X=x]\,dF_{D,X}(d,x)\\
			&=\EE{\big}{\dot{\pi}(D)\partial_d\E[Y|D,X]}, 
		\end{align*}
		where equality (2) follows from the DCT under the regularity assumption: There exist a constant \(c > 0\) such that \(\sup_{|\delta|\leq c} \int  \big|\frac{m(d+\delta,x,e)-m(d,x,e)}{\delta}\big|\,dF_{\varepsilon|D,X}(e|d,x) < \infty\) for every \((d,x)\in\supp{D}\times\supp{X}\). Equality (3) follows from \(\varepsilon \indep D |X\). 
		
		Equality (\ref{eq:id.exo.theta.mu.wooldridge}) can also be derived from equality (\ref{eq:id.exo.theta.mu}). 
		Suppose further that the support $\supp{Y|D,X}$ is bounded and its upper bound  $\overline{y}$ does not depend on $d$. Integration by parts yields that 
		\begin{align*}
			\partial_d \int F_{Y|D,X}(y|d,x)dy
			&=\partial_d \bigg\{ y\,F_{Y|D,X}(y|d,x)\Big|_{y=\underline{y}(d,x)}^{y=\overline{y}(x)} - \int y\,dF_{Y|D,X}(y|d,x)\bigg\} \\ 
			&= \partial_d \Big\{ \overline{y}(x)\cdot 1 - \underline{y}(d,x)\cdot 0 - \E[Y|D=d,X=x]\Big\}\\
			&=-\partial_d\E[Y|D=d,X=x].
		\end{align*}
		Then, applying Proposition 1, we obtain that 
		\begin{align*}
			\theta_\mu
			&=-\int \dot{\pi}(d)\left(\int \partial_d F_{Y|D,X}(y|d,x)\,dy \right) \,dF_{D,X}(d,x) \\
			&=-\int \dot{\pi}(d) \left(\partial_d \int  F_{Y|D,X}(y|d,x)\,dy \right)\,dF_{D,X}(d,x) \\
			&=-\int \dot{\pi}(d) \left(-\partial_d\E[Y|D=d,X=x]\right)\,dF_{D,X}(d,x) \\
			&=\EE{\big}{\dot{\pi}(D)\,\partial_d\E[Y|D,X]}. 
		\end{align*}
		This completes the proof.
	\end{proof}

	\begin{corollaryS}[Identification of \(\theta_\mu\) using the control function]~\\ \label{cory:id.cv.theta.mu}
		For \(\Gamma=\mu\), under the assumptions of Proposition 2 and assuming the support \(\supp{Y}\) is compact, we obtain
		\begin{align*}
			\theta_\mu 
			&= \EE{\big}{\dot{\pi}(D)\,\beta^{\mathrm{CV}}(D,V,X,Y)}\\
			&= \EE{\big}{\dot{\pi}(D)\,\partial_d\E[Y|D,V,X]}
		\end{align*}
		where $\beta^{\mathrm{CV}}(d,v,x,y):=-\frac{\partial_d{F_{Y|D,V,X}(y|d,v,x)}}{f_{Y|D,V,X}(y|d,v,x)}$. 
	\end{corollaryS}
	\begin{proof}[\normalfont\bfseries Proof of Corollary \ref{cory:id.cv.theta.mu}]
		The proof follows arguments analogous to Corollary \ref{cory:id.exo.theta.mu} and is therefore omitted.
	\end{proof}
	
	\subsection{Gini Policy Effects}
	The \textit{Gini coefficient (GC)} is given by
	\[GC:F \mapsto 1 - 2\int_{0}^{1}L_p(F)\,dp, \quad L_p(F):=\frac{\int_{0}^{p}Q_\tau(F)\,d\tau}{\int y\,dF(y)}\] 
	where \(L_p\) denotes the \textit{Lorenz curve}. Then, the Gini MPE is defined as 
	\begin{align*}
		\theta_{GC}:=\lim_{t \downarrow 0}\frac{GC(F_{Y^t}) - GC(F_{Y})}{t}.
	\end{align*}
	\subsubsection{Structural Interpretation}
	\begin{corollaryS}[Structural characterization of \(\theta_{GC}\)]~\\ \label{cory:struct.theta.gini}
		For \(\Gamma=GC\), under the assumptions of Theorem 1 and assuming the support \(\supp{Y}\) is compact, we obtain
		\begin{align*}
			\theta_{GC} = \EE{\Big}{\omega^{GC}(Y,D)\cdot\partial_dm(D,X,\varepsilon)},
		\end{align*}
		where \(\omega^{GC}(y,d):=\frac{2\varphi(F_Y(y))\dot{\pi}(D)}{\E[Y]^2}\) and \(\varphi(\tau) = \int_{0}^{1}(\tau-p)Q_p(F_Y)\,dp\).
	\end{corollaryS}
	\begin{proof}[\normalfont\bfseries Proof of Corollary \ref{cory:struct.theta.gini}]
		We first derive the Hadamard derivative of \(GC\). Under the assumptions given in the corollary\footnote{\citet[Proposition 2]{bhattacharya2007} provides alternative sufficient conditions for the Hadamard differentiability of the Lorenz curve $L_p$ that do not require $Q_\tau(F)$ to be Hadamard differentiable.}, applying Theorem 20.9 (Chain rule) of \citet{vaart2000}, we have
		\begin{align*}
			[L_p]'_{F_Y}(h)
			&=[L_p]'_{(\int_{0}^{p}Q_\tau(F_Y)\,d\tau,\, \mu(F_Y))}\circ \left(\Big[\int_{0}^{p}Q_\tau(F_Y)\,d\tau \Big]'_{F_Y}, \mu'_{F_Y}\right)(h)\\
			&=\frac{\Big[\int_{0}^{p}Q_\tau(F_Y)\,d\tau \Big]'_{F_Y}(h)\cdot\mu(F_Y) - \mu'_{F_Y}(h)\cdot \int_{0}^{p}Q_\tau(F_Y)\,d\tau}{\mu^2(F_Y)}\\
			&=\frac{\int_{0}^{p}[Q_\tau]'_{F_Y}(h) \,d\tau}{\mu(F_Y)} - \frac{L_p(F_Y)}{\mu(F_Y)}
			\int_{0}^{1} [Q_\tau]'_{F_Y}(h) \,d\tau.
		\end{align*}
		where \([Q_\tau]'_{F_Y}(h)= -h(q_\tau)/f_Y(q_\tau)\). 
		Then, \(GC\) is Hadamard differentiable at \(F_Y\) tangentially to \(C(\supp{Y})\), with derivative given by 
		\begin{align}
			GC'_{F_Y}(h)
			&= -2 \int_{0}^{1} [L_p]'_{F_Y}(h)\,dp \notag\\
			&= \frac{2}{\mu(F_Y)}\left\{C_L\int_{0}^{1} [Q_\tau]'_{F_Y}(h) \,d\tau 
			- \int_{0}^{1}\int_{0}^{p}[Q_\tau]'_{F_Y}(h)\,d\tau\,dp \right\} \notag\\
			&\eq{1}{=} \frac{2}{\mu(F_Y)}\left\{C_L \int_{0}^{1} [Q_\tau]'_{F_Y}(h) \,d\tau 
			- \int_{0}^{1}\int_{\tau}^{1}\,dp\,[Q_\tau]'_{F_Y}(h)\,d\tau\right\} \notag\\
			&=\frac{2}{\mu(F_Y)}\int_{0}^{1}\big(C_L - 1 + \tau \big) \left(-\frac{h(q_\tau)}{f_Y(q_\tau)}\right) \,d\tau\notag\\
			&\eq{2}{=} \frac{2}{\mu^2(F_Y)}\int_{0}^{1} \varphi(\tau) \left(-\frac{h(q_\tau)}{f_Y(q_\tau)}\right)\,d\tau\notag
		\end{align} 
		where \(C_L := \int_{0}^{1}L_p(F_Y)\,dp\) and \(\varphi(\tau):=\int_{0}^{1}(\tau - p)Q_p(F_Y) \,dp\). Equality (1) follows from interchanging the order of integration; equality (2) holds because
		\begin{align*}
			C_L - 1 + \tau
			&=\frac{\int_{0}^{1}\int_{0}^{p}Q_\tau(F_Y)\,d\tau\,dp}{\mu(F_Y)} - 1 + \tau\\
			&=\frac{\int_{0}^{1}(1-\tau)Q_\tau(F_Y)\,d\tau}{\mu(F_Y)} - 1 + \tau\\
			&= \frac{\int_{0}^{1}(\tau - p)Q_p(F_Y) \,dp}{\mu(F_Y)}.
		\end{align*} 
		
		Finally, applying Theorem 1, we have
		\begin{align*}
			\theta_{GC} 
			&= GC'_{F_Y}\left(\EEc{\big}{\omega^f(\cdot,D)\cdot\partial_dm(D,X,\varepsilon)}{Y=\cdot}\right)\\
			&= \int_{0}^{1} \EEc{\bigg}{ \frac{2\varphi(\tau)\dot{\pi}(D)}{\E[Y]^2}\cdot\partial_dm(D,X,\varepsilon)}{Y=q_\tau} \,d\tau\\
			&\eq{3}{=} \int_{\supp{Y}} \EEc{\bigg}{ \frac{2\varphi(F_Y(y))\dot{\pi}(D)}{\E[Y]^2}\cdot\partial_dm(D,X,\varepsilon)}{Y=y} \,dF_Y(y)\\
			&=\EE{\bigg}{\frac{2\varphi(F_Y(Y))\dot{\pi}(D)}{\E[Y]^2}\cdot\partial_dm(D,X,\varepsilon)}
		\end{align*}
		where \(\E[Y]=\mu(F_Y)\) and equality (3) follows from the change of variable \(\tau = F_Y(y)\). This completes the proof.
	\end{proof}

	\subsubsection{Identification}
	\begin{corollaryS}[Identification of \(\theta_{GC}\) under conditional independence]~\\ \label{cory:id.exo.theta.gini}
		For \(\Gamma=GC\), under the assumptions of Proposition 1 and assuming the support \(\supp{Y}\) is compact, we obtain
		\begin{align*}
			\theta_{GC} 
			&= \EE{\Big}{\omega^{GC}(Y,D)\cdot \beta(D,X,Y)}\\
			&= \frac{2}{\E[Y]^2}\int_{0}^{1} \varphi(\tau)\cdot\beta^{\mathrm{UQR}}(\tau;\dot{\pi}) \,d\tau, 
		\end{align*}
		where \(\beta^{\mathrm{UQR}}(\tau;\dot{\pi}):=-\frac{\E[\dot{\pi}(D) \partial_dF_{Y|D,X}(q_\tau|D,X)]}{f_Y(q_\tau)}\).
	\end{corollaryS}
	\begin{proof}[\normalfont\bfseries Proof of Corollary \ref{cory:id.exo.theta.gini}]
		Applying Proposition 1, we have
		\begin{align*}
			\theta_{GC}
			&=GC'_{F_Y}\left(\EEc{\big}{\omega^f(\cdot,D)\cdot\beta(D,X,\cdot)}{Y=\cdot}\right)\\
			&=\int_{0}^{1}\EEc{\bigg}{ \frac{2\varphi(\tau)\dot{\pi}(D)}{\E[Y]^2} \cdot\beta(D,X,q_\tau)}{Y=q_\tau}\,d\tau\\
			&=\int_{\supp{Y}}\EEc{\bigg}{\frac{2\varphi(F_Y(y))\dot{\pi}(D)}{\E[Y]^2}\cdot\beta(D,X,y)}{Y=y}\,dF_Y(y)\\
			&=\EE{\bigg}{\frac{2\varphi(F_Y(Y))\dot{\pi}(D)}{\E[Y]^2}\cdot\beta(D,X,Y)}\\
			&=\frac{2}{\E[Y]^2}\int_{\supp{Y}}\varphi(F_Y(y))\iint \,\dot{\pi}(d)\cdot \left(-\frac{\partial_dF_{Y|D,X}(y|d,x)}{f_{Y|D,X}(y|d,x)}\right)\,f_{Y,D,X}(y,d,x)\,dd\,dx\,dy\\
			&=\frac{2}{\E[Y]^2}\int_{\supp{Y}}\varphi(F_Y(y))\int \,\dot{\pi}(d)\cdot \left(-\frac{\partial_dF_{Y|D,X}(y|d,x)}{f_Y(y)}\right)\,dF_{D,X}(d,x)\,dF_Y(y)\\
			&= \frac{2}{\E[Y]^2}\int_{0}^{1} \varphi(\tau)\cdot \left(-\frac{\int\dot{\pi}(d) \partial_dF_{Y|D,X}(q_\tau|d,x)\,dF_{D,X}(d,x)}{f_Y(q_\tau)}\right) \,d\tau, 
		\end{align*}
		where the third and final equalities follow from the change of variable \(\tau=F_Y(y)\). This completes the proof.
	\end{proof}
	
	\begin{corollaryS}[Identification of \(\theta_{GC}\) using the control function]~\\ \label{cory:id.cv.theta.gini}
		For \(\Gamma=GC\), under the assumptions of Proposition 2 and assuming the support \(\supp{Y}\) is compact, we obtain
		\begin{align*}
			\theta_{GC} 
			&= \EE{\Big}{\omega^{GC}(Y,D)\cdot \beta(D,V,X,Y)}\\
			&= \frac{2}{\E[Y]^2}\int_{0}^{1} \varphi(\tau)\cdot \tilde{\beta}^{\mathrm{UQR}}(\tau;\dot{\pi}) \,d\tau.
		\end{align*}
		where \(\tilde{\beta}^{\mathrm{UQR}}(\tau;\dot{\pi}):=-\frac{\E[\dot{\pi}(D) \partial_dF_{Y|D,V,X}(q_\tau|D,V,X)]}{f_Y(q_\tau)}\).
	\end{corollaryS}
	\begin{proof}[\normalfont\bfseries Proof of Corollary \ref{cory:id.cv.theta.gini}]
		The proof follows arguments analogous to Corollary \ref{cory:id.exo.theta.gini} and is therefore omitted.
	\end{proof}
	
\end{appendices}

\singlespacing



\begin{thebibliography}{}
	
	\bibitem[Alejo et al.(2024)]{alejo2024}
	Alejo, J., Galvao, A. F., Mart\'inez-Iriarte, J., and Montes-Rojas, G. (2024).
	``Unconditional quantile partial effects via conditional quantile regression.''
	\textit{Journal of Econometrics}, 105678.
	
	\bibitem[Baltagi and Ghosh(2019)]{baltagi2019}
	Baltagi, B. H., and Ghosh, P. K. (2019). 
	``Partial distributional policy effects under endogeneity.''
	\textit{Sankhya B}, 81(Suppl 1), 123-145.
	
	\bibitem[Chernozhukov et al.(2013)]{cherno2013}
	Chernozhukov, V., Fern\'andez-Val, I., and Melly, B. (2013).
	''Inference on counterfactual distributions.''
	\textit{Econometrica}, 81(6), 2205-2268.
	
	\bibitem[Chernozhukov et al.(2015)]{chernoPanel2015}
	Chernozhukov, V., Fern\'andez-Val, I., Hoderlein, S., Holzmann, H., and Newey, W. (2015). 
	``Nonparametric identification in panels using quantiles.''
	\textit{Journal of Econometrics}, 188(2), 378-392.
	
	\bibitem[Chernozhukov et al.(2018)]{ddml2018}
	Chernozhukov, V., Chetverikov, D., Demirer, M., Duflo, E., Hansen, C., Newey, W., and Robins, J. (2018).
	``Double/debiased machine learning for treatment and structural parameters.'' 
	\textit{The Econometrics Journal}, 21(1), C1-C68.
	
	\bibitem[Chernozhukov et al.(2022a)]{adml2022}
	Chernozhukov, V., Newey, W. K., and Singh, R. (2022a). 
	``Automatic debiased machine learning of causal and structural effects.'' 
	\textit{Econometrica}, 90(3), 967-1027.
	
	\bibitem[Chernozhukov et al.(2022b)]{localRobust2022}
	Chernozhukov, V., Escanciano, J. C., Ichimura, H., Newey, W. K., and Robins, J. M. (2022b). 
	``Locally robust semiparametric estimation.'' 
	\textit{Econometrica}, 90(4), 1501-1535.
	
	
	\bibitem[Fang and Santos(2019)]{fang2019}
	Fang, Z., and Santos, A. (2019). 
	``Inference on directionally differentiable functions.''
	\textit{The Review of Economic Studies}, 86(1), 377-412.
	
	\bibitem[Firpo et al.(2009)]{firpo2009}
	Firpo, S., Fortin, N. M., and Lemieux, T. (2009),
	``Unconditional quantile regressions.''
	\textit{Econometrica}, 77 (3), 953--973.
	
	\bibitem[Firpo and Pinto(2016)]{firpo2016}
	Firpo, S., and Pinto, C. (2016). 
	``Identification and estimation of distributional impacts of interventions using changes in inequality measures.'' 
	\textit{Journal of Applied Econometrics}, 31(3), 457-486.
	
	
	\bibitem[Hoderlein and Mammen(2007)]{hm2007}
	Hoderlein, S., and Mammen, E. (2007).
	''Identification of marginal effects in nonseparable models without monotonicity.''
	\textit{Econometrica}, 75(5), 1513-1518.
	
	
	\bibitem[Hoderlein and Mammen(2009)]{hm2009}
	Hoderlein, S., and Mammen, E. (2009).
	``Identification and estimation of local average derivatives in nonseparable models without monotonicity.''
	\textit{The Econometrics Journal}, 12(1), 1-25.
	
	
	\bibitem[Imbens and Newey(2009)]{imbensNewey2009}
	Imbens, G. W., and Newey, W. K. (2009). 
	``Identification and estimation of triangular simultaneous equations models without additivity.''
	\textit{Econometrica}, 77 (5), 1481--1512.
	
	\bibitem[Mart\'inez-Iriarte et al.(2024)]{mms2024}
	Mart\'inez-Iriarte J, Montes-Rojas G, Sun Y (2024),
	``Unconditional effects of general policy interventions.''
	\textit{Journal of Econometrics}, 238(2): 105570.
	
	\bibitem[Rothe(2010a)]{rothe2010joe}
	Rothe, C. (2010a). 
	``Nonparametric estimation of distributional policy effects.'' 
	\textit{Journal of Econometrics}, 155(1), 56-70.
	
	\bibitem[Rothe(2010b)]{rothe2010el}
	Rothe, C. (2010b),
	``Identification of unconditional partial effects in nonseparable models.''
	\textit{Economics Letters}, 109 (3), 171--174.
	
	\bibitem[Rothe(2012)]{rothe2012}
	Rothe, C. (2012),
	``Partial distributional policy effects.''
	\textit{Econometrica}, 80 (5), 2269--2301.
	
	
	\bibitem[Sasaki(2015)]{sasaki2015}
	Sasaki, Y. (2015). 
	``What do quantile regressions identify for general structural functions?'' 
	\textit{Econometric Theory}, 31(5), 1102-1116.
	
	
	\bibitem[Sasaki et al.(2022)]{sasaki2022}
	Sasaki, Y., Ura, T., and Zhang, Y. (2022). 
	``Unconditional quantile regression with high-dimensional data.'' 
	\textit{Quantitative Economics}, 13(3), 955-978.
	
	
	\bibitem[Su et al.(2019)]{su2019}
	Su, L., Ura, T., and Zhang, Y. (2019). 
	``Nonseparable models with high-dimensional data.''
	\textit{Journal of Econometrics}, 212(2), 646-677.
	
	\bibitem[van der Vaart(2000)]{vaart2000}
	van der Vaart, AW (2000),
	``Asymptotic Statistics.''
	Cambridge University Press
	
	
	\bibitem[Wooldridge(2004)]{wooldridge2004}
	Wooldridge, J. M. (2004). 
	``Estimating average partial effects under conditional moment independence assumptions (No. CWP03/04).''
	cemmap working paper.
	
\end{thebibliography}

\begin{thebibliography}{}
		
		
		\bibitem[Bhattacharya(2007)]{bhattacharya2007}
		Bhattacharya, D. (2007). 
		``Inference on inequality from household survey data.'' 
		\textit{Journal of Econometrics}, 137(2), 674-707.
		
		
		
		\bibitem[Fan and Patton(2014)]{fanCopula2014}
		Fan, Y., and Patton, A. J. (2014). 
		``Copulas in econometrics.'' Annu. Rev. Econ., 6(1), 179-200.
		
		\bibitem[Hoderlein and Mammen(2007)]{hm2007}
		Hoderlein, S., and Mammen, E. (2007).
		''Identification of marginal effects in nonseparable models without monotonicity.''
		\textit{Econometrica}, 75(5), 1513-1518.
		
		\bibitem[Hoderlein and Mammen(2009)]{hm2009}
		Hoderlein, S., and Mammen, E. (2009).
		``Identification and estimation of local average derivatives in nonseparable models without monotonicity.''
		\textit{The Econometrics Journal}, 12(1), 1-25.
		
		
		
		
		\bibitem[Rothe(2010)]{rothe2010el}
		Rothe (2010),
		``Identification of unconditional partial effects in nonseparable models.''
		\textit{Economics Letters}, 109 (3), 171--174.
		
		
		\bibitem[Rothe(2012)]{rothe2012}
		Rothe (2012),
		``Partial distributional policy effects.''
		\textit{Econometrica}, 80 (5), 2269--2301.
		
		
		\bibitem[van der Vaart(2000)]{vaart2000}
		van der Vaart, AW (2000),
		``Asymptotic Statistics.''
		Cambridge University Press
		
		
		\bibitem[Wooldridge(2004)]{wooldridge2004}
		Wooldridge, J. M. (2004).
		``Estimating average partial effects under conditional moment independence assumptions.''
		No. CWP03/04, cemmap working paper.

	
\end{thebibliography}
\end{document}